\documentclass[5p]{elsarticle}
\usepackage{flushend} 
\usepackage{romannum}
\usepackage{mathtools}
\usepackage{cuted}
\usepackage{romannum}
\usepackage[usenames,dvipsnames]{xcolor}
\usepackage{tabularx,booktabs}
\usepackage{setspace}
\usepackage{epstopdf}

\newcommand{\boundellipse}[3]
{(#1) ellipse (#2 and #3)
}

\usepackage{setspace}
\usepackage{calrsfs}
\DeclareMathAlphabet{\pazocal}{OMS}{zplm}{m}{n}
\usepackage{graphicx,dblfloatfix}
\usepackage{graphics} 
\usepackage{epsfig} 
\usepackage{mathptmx}
\usepackage{caption}
\usepackage{subcaption}
\usepackage{times} 
\usepackage{tikz}
\usetikzlibrary{positioning}
\usepackage{amsmath, amssymb}

\usepackage{amsthm}
\usepackage{amsfonts} 
\usepackage{setspace}
\usepackage{scalefnt}
\usepackage{multirow}
\usepackage{textcomp}
\usepackage[english]{babel}
\usepackage{float} 
\usepackage{algorithmicx}
\usepackage[section]{algorithm}
\usepackage{algpascal}
\usepackage{algc}
\usepackage{algcompatible}
\usepackage{algpseudocode}
\let\oldReturn\Return
\renewcommand{\Return}{\State\oldReturn}
\usepackage{mathtools}
\usepackage{bm}
\usepackage{mathptmx}
\usepackage{times} 
\usepackage{mathtools, cuted}
\usepackage{textcomp}
\usepackage[english]{babel}
\usepackage{rotating}
\usepackage{pgfplots}
\pgfplotsset{compat=1.5}
\usepackage{pgfplotstable}
\usepackage{filecontents}
\usepackage{microtype}

\newenvironment{example}
    {\it{Subject to:~}
    \begin{minipage}[t]{0.8\linewidth}}
    {\end{minipage}}
    
\newtheorem{theorem}{Theorem}[section]
\newtheorem{prob}[theorem]{Problem}
\newtheorem{lemma}[theorem]{Lemma}

\newtheorem{assume}[theorem]{Assumption}

\newtheorem{defn}[theorem]{Definition}
\newtheorem{cor}[theorem]{Corollary}

\newtheorem{prop}[theorem]{Proposition}

\newcommand{\A}{\pazocal{A}}

\renewcommand{\S}{\pazocal{S}}

\newcommand{\bp}{{\bf P}}

\newcommand{\D}{\pazocal{D}}

\newcommand{\bS}{{\bf S}}

\newcommand{\sA}{{\scriptscriptstyle{A}}}
\newcommand{\sB}{{\scriptscriptstyle{B}}}
\newcommand{\sD}{{\scriptscriptstyle{D}}}

\newcommand{\G}{{\pazocal{G}}}
\newcommand{\F}{\pazocal{F}}

\graphicspath{{figures/}}

\def \*{\star}
\def \10n{\!\!\!\!\!\!\!\!\!\!}

\graphicspath{{figures/}}
\usepackage[utf8]{inputenc}

         \DeclareMathAlphabet{\mathscr}{U}{BOONDOX-cal}{m}{n}
         \SetMathAlphabet{\mathscr}{bold}{U}{BOONDOX-cal}{b}{n}
         \DeclareMathAlphabet{\mathbscr} {U}{BOONDOX-cal}{b}{n}

\begin{document}\sloppy
\begin{frontmatter}
\title{\LARGE \bf Stochastic Dynamic Information Flow Tracking Game using\\ Supervised Learning for Detecting  Advanced Persistent Threats}
\author[uw]{Shana Moothedath}
\ead{sm15@uw.edu}
\author[uw]{Dinuka Sahabandu}
\ead{sdinuka@uw.edu}
\author[gatech]{Joey Allen}
\ead{jallen309@gatech.edu}
\author[uw]{Linda Bushnell}
\ead{lb2@uw.edu}
\author[gatech]{Wenke Lee}
\ead{wenke@cc.gatech.edu}
\author[gatech]{Radha Poovendran}
\ead{rp3@uw.edu}
\address[uw]{Department of Electrical and Computer Engineering, University of Washington,  USA}
\address[gatech]{College of Computing, Georgia Institute of Technology, USA}
\begin{abstract}
Advanced persistent threats (APTs) are organized  prolonged cyberattacks by sophisticated  attackers with the intend of stealing critical information. Although APT activities are stealthy and evade detection by  traditional detection tools, they interact with the system components  to make progress in the attack.  These interactions lead to information flows that are recorded in the form of a system log.  Dynamic Information Flow Tracking (DIFT) has been proposed as one of the effective ways to detect APTs using the  information flows. A DIFT-based detection mechanism  dynamically performs security analysis on the information flows  to detect possible attacks. However, wide range security analysis using DIFT  results in a significant increase in performance overhead and high rates of false-positives and false-negatives.  In this paper, we model the strategic interaction between APT and DIFT as a non-cooperative stochastic game. The game unfolds on a state space constructed from an information flow graph (IFG) that is extracted from the system log. The objective of the APT in the game is to  choose transitions in the IFG to find an optimal path in the IFG from an entry point of the attack to an attack target. On the other hand, the objective of  DIFT is to dynamically select nodes in the IFG to perform security analysis for detecting APT. Our game model has imperfect information as the players do not have information about the actions of the opponent.  We consider two scenarios of the game~(i)~when the false-positive and false-negative rates  are known to both  players and (ii)~when the false-positive and false-negative rates  are unknown to both  players. Case~(i) translates to a game model with complete information and we propose a value iteration-based algorithm and prove the convergence. Case~(ii) translates to an incomplete information  game with unknown transition probabilities. In this case, we propose Hierarchical Supervised Learning (HSL) algorithm  that integrates a neural network, to predict the value vector of the game, with a policy iteration algorithm to compute an approximate equilibrium. We implemented our algorithms for cases~(i) and~(ii) on real attack datasets for nation state and ransomware attacks and validated the performance of our approach. 
\end{abstract}
\end{frontmatter}
\section{Introduction}\label{sec:intro}

Advanced persistent threats (APTs) are sophisticated  and prolonged cyberattacks that target  specific high value organizations in sectors such as national defense, manufacturing, and the financial industry \cite{jang2014survey, watkins2014impact}. APTs  use advanced attack methods, including advanced exploits of zero-day vulnerabilities, as well as highly-targeted spear phishing and other social engineering techniques to gain access to a network and then remain undetected for a prolonged period of time. The attacking strategies of APTs are stealthy and are methodically designed to bypass conventional security mechanisms  to cause more permanent, significant, and irreversible damages to the network. 

Interaction of APTs with the network introduce information flows in the form of control flow and command flow commands which get recorded in the system log. Analyzing information flows is thus a promising approach to detect presence of APTs \cite{ji2017rain}.
 Dynamic information flow tracking (DIFT) \cite{suh2004secure}  is a flow tracking-based mechanism that is widely used to detect APTs. DIFT tags information flows originating from suspicious input channels in the network and  tracks the propagation of the tagged flows. DIFT then initiates security check points, referred to as {\em traps}, to analyze the tagged flows and verify the authenticity of the flows. As the ratio of the malicious flows to the benign flows is very small in the network, implementation of DIFT introduces significant performance  overhead  on the network  and 
 generates false-positives and false-negatives \cite{ji2017rain}.
 
 This paper models detection of APTs using a DIFT-based detection mechanism. The effectiveness of detection depends on both the security policy of the DIFT and also on the actions of the APT. We model the strategic interaction between APTs and DIFT as a stochastic game. The game unfolds on a state space that is defined using the information flow graph (IFG) constructed from the system log of the system. The  objective of the DIFT is to select locations in the system, i.e., nodes in the IFG, to place the security check points while minimizing the errors due to generation of false-positives and false-negatives.  On the other hand, the objective of the APT is to choose an attack path starting from an entry point to a target node in the IFG. We note that, while APTs aim to achieve a reachability objective, the aim of DIFT is an optimal node selection. The APT-DIFT game has imperfect information structure as APTs are unaware of the locations at which DIFT places security check points and DIFT is unaware of the path chosen by the APT for the attack. We consider two scenarios of the game: (i)~when the false-positive and false-negative rates are known to both players and (ii)~when the  false-positive and false-negative rates are unknown to both players. While case~(i) is a complete information game,  case~(ii) is an incomplete information game with unknown transition probabilities.

We make the following  contributions in this paper:
\begin{enumerate}
\item[$\bullet$]We formulate a constant-sum stochastic game with total reward structure that models the interaction between the DIFT  and APT. The game  captures the information asymmetry between the attacker and the defender, false-positives, and false-negatives generated by DIFT. 
\item[$\bullet$] When the transition probabilities are known, we present a value iteration algorithm to compute equilibrium strategy and prove convergence and polynomial-time complexity.
\item[$\bullet$] When the transition probabilities are unknown, we propose HSL (Hierarchical Supervised Learning), a supervised learning-based algorithm  to compute an approximate equilibrium strategy of the incomplete information game. HSL  integrates a neural network with  policy iteration  and utilizes the hierarchical structure of the state space of the APT-DIFT game to guarantee convergence.
\item[$\bullet$] We implement our algorithms on attack data, of a nation state attack and a ransomware attack, collected using Refinable Attack INvestigation  (RAIN) system, and  show the convergence of the algorithm. 
\end{enumerate}


This paper is organized is as follows: Section~\ref{sec:Related work} summarizes the related work.  Section~\ref{sec:pre} elaborates on information flow graph, and the  attacker and defender models. Section~\ref{sec:game} presents the formulation of the  stochastic game between APT and DIFT. Section~\ref{sec:solution-concept} discusses the min-max solution concept used for solving the game.  Section~\ref{sec:results} presents the value iteration algorithm for the model-based approach and  the  supervised learning-based approach for the model-free case. Section~\ref{sec:sim}  illustrates the results and discussions on the experiments conducted. Section~\ref{sec:end} concludes the paper.

\section{Related work}\label{sec:Related work}


Stochastic games model the  strategic interactions among multiple agents or players in  dynamical systems and  are well studied in the context of  security games  \cite{LyeWin-05},  economic games \cite{Ami-03}, and resilience of cyber-physical systems \cite{ZhuBas-11}. While stochastic games provide a strong framework for modeling security problems, often solving stochastic games are hard. There exists dynamic programming-based approaches,  value-iteration and policy iteration, for computing Nash equilibrium (NE) strategies of these games. However, in many stochastic game formulations  dynamic programming-based approaches do not converge.
 References  \cite{Ahmadi-1, Ahmadi-2}  modeled and studied the interaction between defender and attacker in a cyber setting as a two-player partially observable stochastic game and presented an approach for synthesizing sub-optimal strategies.

Multi-agent reinforcement learning (MARL) algorithms have been proposed in the literature to obtain NE strategies of zero-sum  and nonzero-sum stochastic games, when the game information such as transition probabilities  and payoff functions of the players are unknown. However, algorithms with guaranteed convergence are available only for special cases \cite{hu1998multiagent, hu2003nash,  najim2001adaptive, prasad2015two}. A Nash-Q learning algorithm is given in \cite{hu2003nash} that converges to NE of a general-sum game when the NE is unique. The algorithm in \cite{hu2003nash} is guaranteed to converge for discounted games with  perfect information and unique NE. 
A Q-learning  algorithm is presented in \cite{hu1998multiagent} for discounted, constant-sum games with unknown transition probabilities  when the players have perfect information. An adaptive policy approach using a regularized Lagrange function is proposed in \cite{najim2001adaptive} for zero-sum games with unknown transition probabilities and irreducible state space.  A stochastic approximation-based algorithm is presented in  \cite{prasad2015two} for computing NE of discounted general-sum games with irreducible state space.  
 While there exist algorithms that converge to NE of the game for special cases, there are no known algorithms to compute NE of a  general stochastic game with incomplete information structure.

 Detection of APTs are analyzed using game theory  in the literature by modeling the interactions between APT and the detection mechanism as a game \cite{rass}.
 A dynamic game model is given in \cite{huang2019adaptive} to detect APTs  that  can adopt adversarial deceptions. A honeypot game theoretic model is introduced in \cite{honeypot} to address detection of APTs.

There has been some effort  to model the detection  of APTs by a DIFT-based detection mechanism using game-theoretic framework  \cite{MooShaAllVClaBushWenPoo-18_arx},  \cite{dinukaCDC}, \cite{moothedath2020StochasticDIFT}.   Reference \cite{MooShaAllVClaBushWenPoo-18_arx} studied a DIFT model which selects the trap locations in a dynamic manner and proposed a min-cut based solution approach. Detection of APTs when the attack consists of multiple attackers,  possibly with different capabilities,  is studied in  \cite{dinukaCDC}. We note that, the game models in  \cite{MooShaAllVClaBushWenPoo-18_arx} and \cite{dinukaCDC}  did not consider  false-negative and false-postive generation by DIFT and are hence  non-stochastic. Recently, stochastic models  of APT-DIFT games was proposed in \cite{moothedath2020StochasticDIFT}, \cite{SahMooAllClaLeePoo-18}.  Reference  \cite{SahMooAllClaLeePoo-18} proposed a value iteration-based algorithm to obtain an $\epsilon$-NE of the {\em discounted} game and \cite{moothedath2020StochasticDIFT} proposed a min-cut solution approach to compute an NE of the game. Formulations in  \cite{moothedath2020StochasticDIFT}, \cite{SahMooAllClaLeePoo-18} assume that the transition probabilities of the game (false-positives and false-nagatives)   are known. 

DIFT-APT games with unknown transition probabilities are studied in \cite{sahabandu2019stochasticGameSec}, \cite{Sahabandu-arxiv-2020}.  While  reference \cite{sahabandu2019stochasticGameSec} proposed a two-time scale algorithm to compute an equilibrium of the discounted game,  \cite{Sahabandu-arxiv-2020} considered an average reward payoff structure.   Moreover, the game models in \cite{sahabandu2019stochasticGameSec}, \cite{Sahabandu-arxiv-2020} resulted in a unichain structure on the state space of the game, unlike  the game model considered in this paper. The unichain structure of the game  is critically utilized in \cite{sahabandu2019stochasticGameSec}, \cite{Sahabandu-arxiv-2020} for developing the solution approach and deriving the results. Reference \cite{misra2019learning}, the preliminary conference version of this work, studied DIFT-APT game with unknown transition probabilities and  average reward payoff structure, when the state space graph is not a unichain structure. 
The approach in  \cite{misra2019learning} approximated the payoff function to a convex function and utilized an input convex neural network (ICNN) architecture. The ICNN is  integrated with an alternating optimization technique   and  empirical results are presented in \cite{misra2019learning} to learn approximate equilibrium strategies. In this paper, we do not restrict the payoff function to be convex and hence relax the condition for the neural network to be input convex.

\section{Preliminaries}\label{sec:pre}
\subsection{Information Flow Graph (IFG)}\label{sec:IFG}
Information Flow Graph (IFG) provides a graphical representation of  the whole-system execution during the entire period of monitoring.  We use  RAIN recording system \cite{ji2017rain} to construct the IFG. RAIN comprises a kernel module which logs all the system calls that are requested by the user-level processes in the target host. The target host then sends the recorded system recorded logs to the analysis host. The analysis host consists of a provenance graph builder, which then constructs the coarse-grained IFG. The coarse-IFG is then refined using various pruning techniques. A brief discussion on the pruning steps is included in Section~\ref{sec:sim}. For more details, please refer \cite{ji2017rain}.

IFGs are widely used by analysts for  effective cyber response \cite{hossain2018dependence}, \cite{ji2017rain}. IFGs are directed graphs, where the nodes in the graph form entities such as processes, files, network connections, and memory objects in the system. Edges correspond to the system calls and are oriented in the direction of the information flows and/or causality \cite{hossain2018dependence}. Let  directed graph $\G=(V_{\G}, E_{\G})$ represent  IFG of the system, where $V_{\G}=\{v_1, \ldots, v_N\}$ and $E_{\G} \subseteq V_{\G} \times V_{\G}$.  Given a system log, one can build the corase-grained-IFG which is then pruned and refined incrementally to obtain the corresponding  IFG \cite{ji2017rain}. 

\subsection{Attacker Model}\label{sec:attacker}
This paper consider advanced cyberattacks called as  APTs. APTs are information security threats that target sensitive information in specific organizations. 
APTs  enter into the system by leveraging  vulnerabilities in the system and implement multiple sophisticated methods to continuously and stealthily steal information \cite{hossain2018dependence}. A typical APT attack consists of multiple stages  initiated by a successful penetration and followed by initial compromise, C\&C communications, privilege escalation, internal reconnaissance, exfiltration, and cleanup \cite{hossain2018dependence}. Detection of APT attacks are very challenging as the attacker activities blend in seamlessly with normal system operation. Moreover,  as APT attacks are customized, they can not be detected using  signature-based detection methods such as firewalls, intrusion detection systems, and antivirus software. 

\subsection{Defender Model}\label{sec:DIFT}
Dynamic information flow tracking (DIFT) is a promising technique to detect security attacks on  computer systems \cite{ClaLiOrs:07, EncGilHanTen-14, suh2004secure}. 
DIFT  tracks the system calls to detect the malicious information flows from an adversary and to restrict the use of these malicious flows. DIFT architecture is composed of (i)~tag sources, (ii)~tag propagation rules, and (iii)~tag sinks (traps). 
Tag sources are the suspicious locations in the system, such as  keyboards, network interface, and hard disks,  that are tagged/tainted as suspicious by DIFT.  Tags are single bit or multiple bit markings depending on the level of granularity manageable with the available memory and resources. All the processed values of the tag sources are tagged and DIFT tracks the propagation of the tagged flows. 
When anomalous behavior is detected in the system, DIFT initiates security analysis and tagged flows are  inspected by DIFT at specific locations, referred to as {\em traps}.  DIFT conducts fine grain analysis at traps to detect the attack and to  perform risk assessment. 
 While tagging and trapping using DIFT is a promising detection mechanism against APTs, DIFT introduces  performance overhead on the system. Performing security  analysis (trapping) of tagged flows uses considerable amount of memory of the system \cite{EncGilHanTen-14}.  Thus there is a tradeoff between system log granularity and  performance \cite{ji2017rain}. 

\section{Problem Formulation}\label{sec:game}
 
This section formulates the interaction between the APT and the DIFT-based defense mechanism as a stochastic game where the decisions taken by the adversary (APT) and the defender (DIFT), referred to as agents/players, influences the system behavior. The objective of the adversary is to choose transitions in the IFG so as to reach the destination node set  $\D \subset V_{\G}$ from the set of entry points  $\lambda \subset V_{\G}$. On the other hand, the objective of the defender is to dynamically choose locations to trap the information flow so as to secure the system from any possible attack. 
The state of the system denotes the location of the tagged information flow in the system.    
The defender cannot distinguish a malicious and a benign flow. 
On the other hand, the adversary does not know if the tagged flow will get trapped by the defender while choosing a transition. Thus the game is an imperfect information game. We note that, the state of the game is known here unlike in a partially observable game setting where the state of the game is unknown.
The system's current state and the joint actions of the agents together determine a probability distribution over the  possible  next states of the system.

\subsection{State Space}
The state of the game at a time step $t$, denoted as $s_t$, is defined as the position of the tagged information flow in the system.  Let  $\bS=\{v_0, v_1, \ldots, v_N, \phi, \tau_{\sA}, \tau_{\sB}\}$ be the finite state space of the game. Here, $s_t = \phi$ denotes the state when the tagged flow drops out by abandoning the attack, $s_t = \tau_{\sA}$ denotes the state when DIFT detects the adversary after performing the security analysis, and $s_t = \tau_{\sB}$ denotes the state when DIFT  performs  security analysis and concludes a benign flow as malicious (false positive). Let $\D$ be the destination (target) node set of the adversary and $|\D|=q$.   Without loss of generality, let $\D = \{v_{1}, v_{2}, \ldots,  v_q\}$. We assume that both agents know the IFG and the destination set  $\D$. 
The state $v_0$ corresponds to a virtual state that denote the starting point of the game.  In the state space $\bS$, $v_0$ is connected to all nodes in $\lambda$. Thus   the state  of the game at $t=0$ is  $s_0= v_0$.

\subsection{Action Spaces}
At every time instant in the game, the players choose their actions from their respective action sets.  The action set of the adversary is defined as the possible set of transitions the adversarial flow can execute at a state. The defender has limited resources and hence can perform security analysis  at one information flow at a time. Thus the defender's action set is defined such that, at every decision point, the defender chooses one node to perform security analysis, among the possible nodes that the adversary can transition to.

Let the action set of the adversary and the defender at a state $s_t$ be  $\A_{\sA}(s_t)$ and $\A_{\sD}(s_t)$, respectively. At a state $s_t\in V_{\G}$, the adversary chooses an action either to drop out, i.e., abort the attack, or to continue the attack by transitioning to a neighboring node. 
Thus $\A_{\sA}(s_t) :=  \{\phi\} \cup \{v_j: s_t=v_i~{\rm and~} (v_i, v_j) \in E_{\G}\} $.  On the other hand, the defender's action set at state $s_t$ $\A_{\sD}(s_t):= \{0\} \cup \{v_j: s_t=v_i~{\rm and~} (v_i, v_j) \in E_{\G}\}$, where $\A_{\sD}(s_t)=0$ represents that the tagged information flow is not trapped and $\A_{\sD}(s_t)=v_j$  represents that defender decides to perform security analysis at the node $v_j \in V_{\G}$ at instant $t$. 

The game originates at $t=0$ with state $s_0=v_0$ with $\A_{\sA}(v_0) := \lambda$ and $\A_{\sD}(v_0) :=0$. The definition of action sets at $v_0$  captures the fact that the adversarial flow originates at  one of the  entry points.  Performing security analysis at entry points can not detect the attack as there are not enough traces to analyze. 
Moreover, the  game terminates at time $t$ if the state of the game $s_t \in \{\phi, \tau_{\sA}, \tau_{\sB} \} \cup \D$. The set of states  $\{\phi, \tau_{\sA}, \tau_{\sB}\}\cup \D$ are referred to as the {\em absorbing states} of the  game.  
For $s_t \in \{\phi, \tau_{\sA}, \tau_{\sB}\} \cup \D$, $\A_{\sD}(s_t) = \A_{\sA}(s_t) = \emptyset$. At a non-absorbing state, the players choose their actions from their respective action set and the game evolves until the state of the game is an absorbing state.

\subsection{Transitions}
The transition probabilities of the game are governed by the uncertainty  associated with DIFT. DIFT is not capable of performing security analysis accurately due to the generation of false-positives and false-negatives. 
 Consider a tagged flow incoming at node $v_i \in V_{\G}$ at time $t$. Let the action chosen by the defender and the adversary at $s_t$ be $d_t$ and $a_t$, respectively. If defender chooses not to trap an information flow, then the flow proceeds to the node in $\G$ chosen by the adversary. That is, if $d_t=0$, then $s_{t+1}=a_t$.   If the defender chooses to trap the node at which the adversary also decides to transition to,  then the adversary is detected with probability $1-FN(d_t)$ and the flow transition to the node in $\G$ corresponding to the action of the adversary with the remaining probability. That is, if $d_t=a_t=v_i$, then $s_{t+1}=\tau_{\sA}$ with probability $1-FN(d_t)$ and $s_{t+1}=a_t$ with probability $FN(d_t)$. If the defender decides to trap a node which is not the node the adversary decides to transition to,  then the defender generates a false-positive with probability $FP(d_t)$ and the flow transition to the node in $\G$ corresponding to the action of the adversary with the remaining probability. That is, if $d_t \neq 0$ and $d_t \neq a_t$, then $s_{t+1}=\tau_{\sB}$ with probability $FP(d_t)$ and $s_{t+1}=a_t$ with probability $1-FP(d_t)$. 
 
 Let $P(s_t, s_t, d_t, s_{t+1})$ denote the probability of transitioning to a state $s_{t+1}$ from a state $s_t$ under actions $a_t$ and $d_t$. Then,
\begin{equation}\label{eq:transition_D}
P(s_t, a_t, d_t, s_{t+1}) =
 \begin{cases}
 \begin{array}{lll}
1,  &  s_{t+1}=a_t, & \mbox{~if~} d_t=0\\
FN(d_t), &  s_{t+1}=a_t, & \mbox{~if~} d_t=a_t\\
1-FN(d_t),&   s_{t+1}=\tau_{\sA}, & \mbox{~if~} d_t=a_t\\
FP(d_t),&   s_{t+1}=\tau_{\sB}, & \mbox{~if~} d_t\neq a_t\\
1-FP(d_t),&  s_{t+1}=a_t, & \mbox{~if~} d_t\neq a_t\\
 \end{array}
 \end{cases}
 \end{equation}
 The transition probabilities $FN(v_i), FP(v_i)$, for $v_i \in V_{\G}$, denote the rates of generation of false-negatives and false-positives, respectively, at the different nodes in the IFG. We note that, different nodes in IFG have different capabilities to perform security analysis and depending on that the value of $FN(\cdot)$ and $FP(\cdot)$ are different. The numerical values of these transition probabilities also depend on the type of the attack. As APTs are tailored attacks that can manipulate the system operation and evade conventional security mechanisms such as firewalls, anti-virus software, and intrusion-detection systems,  $FN(\cdot)$'s and $FP(\cdot)$'s  are often  unknown and hard to estimate accurately.
 
\subsection{Strategies}
A strategy for a player is a mapping which yields probability distribution over the player's actions at every state. Consider {\em mixed} (stochastic) and stationary player strategies.   
When the strategy is stationary, the probability of choosing an action at a state depends only  on the current state of the game. 
Let the stationary strategy space of the attacker be $\bp_{\sA}$ and that of the defender be $\bp_{\sD}$. Then, $\bp_{\sA}: \bS \rightarrow [0,1]^{|\A_{\sA}|}$ and $\bp_{\sD}: \bS \rightarrow [0,1]^{|\A_{\sD}|}$. Let $p_{\sA} \in \bp_{\sA}$ and $p_{\sD} \in \bp_{\sD}$.   Here, $p_{\sD}=[p_{\sD}(v_{q+1}), \ldots, p_{\sD}(v_{N})]$, where $p_{\sD}(v_i)$ denotes the probability distribution over all the actions of the defender at state $v_i$, i.e., out-neighboring nodes of $v_i$ in IFG and not trapping. For $p_{\sA}= [p_{\sA}(v_0), p_{\sA}(v_{q+1}), \ldots, p_{\sA}(v_{N})]$, $p_{\sA}(v_i)$ is a probability distribution over all possible out-neighbors of the node  $v_i$ in the IFG and $\phi$.  We note that, $\A_{\sA}(s_t) = \A_{\sD}(s_t)= \emptyset$, for $s_t \in \{\phi, \tau_{\sA}, \tau_{\sB}\} \cup \D$. Also, $\A_{\sD}(v_0)= \emptyset$. 

\subsection{Payoffs}
In the APT-DIFT game, the aim of the APT is to choose transitions in the IFG in order to reach destination by evading detection by DIFT. The aim of DIFT is to select the security check points to secure the system from the APT attack.  Recall that APTs are stealthy attacks that undergo for a long period of time with the ultimate goal of stealing information over a long time. The destructive consequences of APTs are often unnoticeable until the final stages of the attack \cite{BenPekButFel:12}. In this paper we consider the payoff functions of the APT-DIFT game  such that players achieve reward ($\beta$) when their respective aim is achieved. 
 DIFT achieves the aim under two scenarios:~(1)~when the APT is detected successfully and~(2)~when APT drops out the attack. We note that, in both cases~(1) and~(2),  DIFT  secures the system from the APT attack and hence is rewarded $\beta$. On the other hand,  APT achieves the aim under two scenarios:~(i)~when APT reach destination and~(ii)~when a benign flow is concluded as malicious, i.e., false-positive. We note that, when a benign flow is concluded as malicious, DIFT no longer analyzes the flows (as it believes APT is detected) and hence the actual malicious flow evades detection and can achieve the aim. Thus in both cases (i) and~(ii) APT achieves the aim and receives a reward of $\beta$.

Let the payoff of  player $k$ at an absorbing state $s_t$ be $r^k(s_t)$, where $k \in \{A, D\}$.  At state $\tau_{\sA}$ DIFT receives a payoff of $\beta$ and APT receives $0$ payoff.  At state $\tau_{\sB}$ the APT  receives a payoff of $\beta$ and DIFT receives $0$ payoff.  At a state  in the set $\D=\{v_{1}, v_{2}, \ldots, v_q\}$, APT receives a payoff of $\beta$ and DIFT receives $0$ payoff. At state $\phi$ DIFT receives a payoff of $\beta$ and APT receives $0$ payoff.  Then, 
\begin{eqnarray}
r^{\sA}(s_t)\hspace*{-2 mm}&=&\hspace*{-2 mm}\label{eq:Apayoff}
\begin{cases}
\begin{array}{ll}
 \beta,& s_t\in \D\\
 \beta,& s_t=\tau_{\sB}\\
 0,& \mbox{otherwise}
\end{array}
\end{cases}\\
r^{\sD}(s_t)\hspace*{-2 mm}&=&\hspace*{-2 mm}\label{eq:Dpayoff}
\begin{cases}
\begin{array}{ll}
 \beta,& s_t=\tau_{\sA}\\
  \beta,& s_t=\phi\\
 0,& \mbox{otherwise}
\end{array}
\end{cases}
\end{eqnarray}
 At each stage in game, $s_t$ at time $t$, both players simultaneously choose their action $a_t$ and $d_t$ and transition to a next state $s_{t+1}$. This is continued until they reach an absorbing state and receive the payoff defined using Eqs.~\eqref{eq:Apayoff} and~\eqref{eq:Dpayoff}.
 
  Let $T$ denote termination time of the game, i.e., $s_{t+1} = s_{t}$ for $t \geqslant T$.  At the termination time $s_T \in \{\phi, \tau_{\sA}, \tau_{\sB}\}\cup \D$.
Let $U_{\sA}$ and $U_{\sD}$ denote the payoff functions of the adversary and the defender, respectively.  
As the initial state of the game is $v_0$, for a strategy pair $(p_{\sA}, p_{\sD})$ the expected payoffs of the players are 
\begin{eqnarray}
\hspace*{-7.5 mm} U_{\sA}(v_0, p_{\sA}, p_{\sD}) \hspace*{-3 mm}&=&\hspace*{-3 mm} \mathbb{E}_{v_0,p_{\sA}, p_{\sD}} \left[\sum\limits_{t = 0}^{T}r^{\sA}(s_t)\right] = \mathbb{E}_{v_0,p_{\sA}, p_{\sD}} \Big[r^{\sA}(s_T)\Big],\label{eq:Adv_obj}\\
\hspace*{-7.5 mm} U_{\sD}(v_0, p_{\sA}, p_{\sD}) \hspace*{-3 mm}&=&\hspace*{-3 mm} \mathbb{E}_{v_0,p_{\sA}, p_{\sD}} \left[\sum\limits_{t = 0}^{T}r^{\sD}(s_t)\right] = \mathbb{E}_{v_0,p_{\sA}, p_{\sD}} \Big[r^{\sD}(s_T)\Big],\label{eq:Def_obj}
\end{eqnarray}
where $\mathbb{E}_{v_0,p_{\sA}, p_{\sD}}$ denotes the expectation with respect to $ p_{\sA}$ and $p_{\sD}$ when the game originates at state $v_0$. 
The objective of  APT and DIFT is to individually maximize their expected total payoff 
 given in Eq.~\eqref{eq:Adv_obj}  and  Eq.~\eqref{eq:Def_obj}, respectively. 
  Thus the optimization problem solved by DIFT is

\begin{equation}\label{eq:D-prob}
\max_{p_{\sD} \in {\bf P}_{\sD}}U_{\sD}(v_0, p_{\sA}, p_{\sD}).
\end{equation}
Similarly, the optimization problem of the APT is
\begin{equation}\label{eq:A-prob}
\max_{p_{\sA} \in {\bf P}_{\sA}}U_{\sA}( v_0, p_{\sA}, p_{\sD}). 
\end{equation}
 
To bring out the structure of the payoffs well, we let $R_{s_0}(s)$ be the cumulative probability with which  the state of the game at the termination time is $s$,   when the game  originates at $s_0$.  With slight abuse of notation, we use $R_{s_0}(\D)$ to denote the cumulative probability with which  the state of the game at the termination time lies in set $\D$,   when the game  originates at $s_0$.  At the time of termination, i.e., $t=T$, the state of the game satisfies one of the following: (i)~$s_T =\tau_{\sA}$, (ii)~$s_T =\tau_{\sB}$, (iii)~$s_T\in \D$, and (iv)~$s_T =\phi$. 
Using these definitions Eqs.~\eqref{eq:Adv_obj} and~\eqref{eq:Def_obj} can be rewritten as
 \begin{eqnarray}
U_{\sA}(v_0, p_{\sA}, p_{\sD}) &=& \Big( R_{s_0}(\D)+R_{s_0}(\tau_{\sB})\Big)\, \beta, \label{eq:Adv_final}\\
U_{\sD}(v_0, p_{\sA}, p_{\sD}) &=& \Big(R_{s_0}({\tau_{\sA}})+ R_{s_0}({\phi} )\Big) \beta.\label{eq:Def_final}
\end{eqnarray}

Using the reformulation of the payoff functions,
we present the following property of the APT-DIFT game.
\begin{prop}\label{prop:constant-sum}
The APT-DIFT stochastic game is a constant sum game.
\end{prop}
\begin{proof}
Recall that APT-DIFT game has absorbing states $\phi, \tau_{\sA}, \tau_{\sB}$ and $\D$. At  the termination time of the game, i.e., $t=T$, the state of the game $s_T \in \{\phi, \tau_{\sA}, \tau_{\sB} \} \cup \D$. This implies  $$R_{s_0}(\phi)+ R_{s_0}(\tau_{\sA})+ R_{s_0}(\tau_{\sB})+R_{s_0}(\D)=1.$$
 This gives $U_{\sA}(v_0, p_{\sA}, p_{\sD})+U_{\sD}(v_0, p_{\sA}, p_{\sD}) =$ $\Big( R_{s_0}(\phi)+ R_{s_0}(\tau_{\sA})+ R_{s_0}(\tau_{\sB})+R_{s_0}(\D) \Big)\,\beta = \beta$. Thus the game between APT and DIFT is a constant sum game with constant equal to $\beta$.
\end{proof}
\section{Solution Concept}\label{sec:solution-concept}
The solution concept of the game is defined as follows. 
 Each player is interested in maximizing their individual reward in the {\em minimax} sense. In other words, each player is assuming the worst case for an optimal opponent player.
Since the APT-DIFT game is a constant-sum game, it is sufficient to view that the opponent player is acting to minimize the reward of the agent. Thus DIFT chooses its actions to find an optimal strategy $p^{\*}_{\sD}$ that achieves the upper value defined as

\begin{equation}\label{eq:upper-value}
\overline{V}(s_0) = \max_{p_{\sD} \in {\bf P}_{\sD}}\min_{p_{\sA} \in {\bf P}_{\sA}}U_{\sD}(v_0, p_{\sA}, p_{\sD}).
\end{equation}

On the other hand,  APT chooses its actions  to find an optimal strategy $p^{\*}_{\sA}$  that achieves

\begin{equation}\label{eq:A-value}
 \max_{p_{\sA} \in {\bf P}_{\sA}}\min_{p_{\sD} \in {\bf P}_{\sD}}U_{\sA}(v_0, p_{\sA}, p_{\sD}) = \max_{p_{\sA} \in {\bf P}_{\sA}}\min_{p_{\sD} \in {\bf P}_{\sD}}\Big(\beta-U_{\sD}(v_0, p_{\sA}, p_{\sD})\Big).
\end{equation}
This is equivalent to saying that APT aims to find an optimal strategy $p^{\*}_{\sA}$ that achieves the lower  value defined as
\begin{equation}\label{eq:lower-value}
\underbar{V}(s_0) =  \min_{p_{\sA} \in {\bf P}_{\sA}}\max_{p_{\sD} \in {\bf P}_{\sD}}U_{\sD}(v_0,  p_{\sA}, p_{\sD}).
\end{equation}
From Eqs.~\eqref{eq:upper-value} and~\eqref{eq:lower-value}, the defender tries to maximize and the adversary tries to  minimize $U_{\sD}(v_0, p_{\sA}, p_{\sD})$. Hence the value of the game is defined as 
\begin{eqnarray}\label{eq:value-def}
V^{\*}(s_0) &:=&  \max_{p_{\sD} \in {\bf P}_{\sD}}U_{\sD}(v_0, p^{\*}_{\sA}, p_{\sD})= U_{\sD}(v_0, p^{\*}_{\sA}, p^{\*}_{\sD})\nonumber\\ 
&=&  \min_{p_{\sA} \in {\bf P}_{\sA}}U_{\sD}(v_0, p_{\sA}, p^{\*}_{\sD}).
\end{eqnarray}
The strategy pair $(p^{\*}_{\sA}, p^{\*}_{\sD})$ is referred to as the saddle point or Nash equilibrium (NE) of the game.

\begin{defn}\label{def:epsilon-NE}
Let $(p^{\*}_{\sA}, p^{\*}_{\sD})$ be a Nash equilibrium of the APT-DIFT game. A strategy pair $(p_{\sA}, p_{\sD})$ is said to be an $\epsilon-$Nash equilibrium, for $\epsilon > 0$, if $$U_{\sD}(v_0, p_{\sA}, p_{\sD}) \geqslant  U_{\sD}(v_0,  p^{\*}_{\sA}, p^{\*}_{\sD})-\epsilon.$$  In other words, the value corresponding to the strategy pair $(p_{\sA}, p_{\sD})$ denoted as $V(s_0)$ satisfies $$V(s_0) \geqslant V^{\*}(s_0)-\epsilon.$$
\end{defn}

Proposition~\ref{prop:NE} proves the existence of NE for the APT-DIFT game.

\begin{prop}\label{prop:NE}
There exists a  Nash equilibrium for the APT-DIFT  stochastic game.
\end{prop}
Proof of Proposition~\ref{prop:NE} is presented in the appendix.

In the next section we present our approach to compute an NE of the APT-DIFT game.

\section{Solution to the APT-DIFT Game}\label{sec:results}
This section  presents our approach to compute an NE of the APT-DIFT game. Firstly, we propose a model-based approach to compute NE using a value iteration algorithm. Later, we present a model-free approach based on a policy iteration algorithm, when the transition probabilities, i.e., the rate of false-positives and false-negatives, are unknown.

\subsection{Value Iteration Algorithm}

This subsection  presents our solution approach for solving the APT-DIFT game when the false-positive and false-negative rates (transition probabilities) are known. 
By Proposition~\ref{prop:NE}, there exists an NE for the APT-DIFT  game.
Our approach to compute NE of the  APT-DIFT  game is presented below.

Let $s\in {\bf S}$ be an arbitrary state of the APT-DIFT game. The state value function for a constant-sum game, analogous to the Bellman equation, can be written as 
\begin{equation}\label{eq:value}
V^{\*}(s) = \max_{p_{\sD}(s)\in {\bf P}_{\sD}(s)}\min_{a\in\A_{\sA}(s)}\sum_{d\in \A_{\sD}(s)}Q^{\*}(s,a,d)\,p_{\sD}(s,d),
\end{equation}

where the $Q$-values are defined as
\begin{equation}\label{eq:Q-value}
Q^{\*}(s,a,d) = \sum_{s'\in {\bf S}}P(s,a,d,s')V^{\*}(s').
\end{equation}

The min in Eq.~\eqref{eq:value} can also be defined over mixed (stochastic) policies, but, since it is `inside' the max, the minimum is achieved for a deterministic action choice \cite{littman2001value}.  The  strategy selected by Eq.~\eqref{eq:value} is referred to as the {\em minimax} strategy, and given by
\begin{equation}\label{eq:policy}
p^{\*}_{\sD}(s) =\arg\max_{p_{\sD}(s)\in {\bf P}_{\sD}(s)}\min_{a\in\A_{\sA}(s)}\sum_{d\in \A_{\sD}(s)}Q^{\*}(s,a,d)\,p_{\sD}(s,d).
\end{equation}
The aim here is to compute minimax  strategy $p^{\*}_{\sD}$. Our proposed algorithm  and  convergence proof  is given below.

Let $V = \{V(s): s \in {\bf S}\}$ denote the value vector corresponding to a strategy pair $(p_{\sA}, p_{\sD})$. The value for a state  $s \in {\bf S}$, $V(s)$,  is the expected payoff under strategy pair $(p_{\sA}, p_{\sD})$ if the game originates at state $s$. Then, $V^{\*}(s)$ is the  expected payoff corresponding to an NE strategy pair $(p^{\*}_{\sA}, p^{\*}_{\sD})$ if the game originates at state $s$.
 Algorithm~\ref{alg:value}  presents  the pseudocode to compute the value vector of the APT-DIFT game. 
 
 Algorithm~\ref{alg:value} is a value iteration algorithm defined on a value vector $V=\{V(s):s \in {\bf S}\}$,  where $V(s)$ is the value of the game  starting at the state $s$. Thus $V(s)=\beta$, for $s\in \{\phi, \tau_{\sA}\}$, and $V(s')=0$, for $s' \in \{\tau_{\sB}\} \cup \D$. The values of the states is computed recursively, for every iteration $k=1,2,\ldots$, $V^{(k)}(s)$,  as follows:

\begin{equation}\label{eq:val-iterate}
V^{(k)}(s)=
\begin{cases}
 \begin{array}{ll}
\beta, & \mbox{~if~} s\in \{\phi, \tau_{\sA}\},\\
0, & \mbox{~if~} s \in \{\tau_{\sB}\}\cup \D, \\
\mbox{val}(s, V^{(k-1)}), & \mbox{~otherwise},
\end{array}
\end{cases}
\end{equation}
where $\mbox{val}(s, V^{(k-1)})=$  $$ \max\limits_{p_{\sD} \in {\bf P}_{\sD}}\min\limits_{a \in \A_{\sA}(s)}\sum\limits_{s' \in {\bf S}}\sum\limits_{d \in \A_{\sD}(s)}p_{\sD}(s, d)P(s,a,d,s')V^{(k-1)}(s').$$

The value-iteration algorithm computes a sequence $V^{(0)},$ $V^{(1)},$ $V^{(2)}, \ldots$, where for $k=0,1,2\ldots$, each valuation $V^{(k)}$ associates with each state $s \in {\bf S}$ a lower bound $V^{(k)}(s)$ on the value of the game.   Algorithm~\ref{alg:value} need not terminate in
finitely many iterations \cite{real-analysis}. The parameter $\delta$, in step~\ref{step:del}, denotes the acceptable error bound which is the maximum absolute  difference in the values corresponding to two consecutive iterations, i.e., $\max\{|V^{(k)}(s)=V^{(k+1)}(s)|:s \in {\bf S}\}$, and serves as a stopping criteria for Algorithm~\ref{alg:value}. Smaller value of $\delta$ in Algorithm~\ref{alg:value} will return  a value vector that is closer  to the actual value vector. Below we prove that as $k$ approaches $\infty$, the values $V^{(k)}(s)$ approaches the exact values $V(s)$ from below, i.e., $\lim_{k\rightarrow \infty}V^{(k)}(s)$ converges to the value of the game at state $s$. Theorem~\ref{th:value-conv} proves the asymptotic convergence of the values.

 \begin{algorithm}[t]
  \caption{Value iteration algorithm to find NE  strategy in APT-DIFT game\label{alg:value}}
  \begin{algorithmic}
\State \textit {\bf Input:} APT-DIFT game with state space ${\bf S}$, destination set $\D$, action space $\A_{\sD}, \A_{\sA}$, payoff parameter $\beta$, transition probabilities $FN(\cdot), FP(\cdot)$
\State \textit{\bf Output:} Value vector  $\hat{V}$ and  defender policy $\hat{p}_{\sD}$
\end{algorithmic}
  \begin{algorithmic}[1]
  \State Initialize value vector $V^{(0)}(s) \leftarrow 0$, for all $s \in {\bf S}$ and $V^{(1)}(s') \leftarrow \beta$ for $s' \in \{\tau_{\sA}, \phi\}$, $V^{(1)}(s'') \leftarrow 0$ for all $s'' \in{\bf S}\setminus \{\tau_{\sA}, \phi\}$, $k\leftarrow 0$, $\delta \geqslant 0$\label{step:init}
  \While{$\max\{|V^{(k+1)}(s)-V^{(k)}(s)|:s\in {\bf S}\} > \delta$}\label{step:del}
  \State $k \leftarrow k+1$
  \For {$s \notin  \{\phi,  \tau_{\sA}, \tau_{\sB}\}\cup \D$}
  \State {\scalefont{1}{$V^{(k+1)}(s) \leftarrow $ 
  \State $\max\limits_{p_{\sD} \in {\bf P}_{\sD}}\min\limits_{a \in \A_{\sA}(s)}\sum\limits_{s' \in {\bf S}}\sum\limits_{d \in \A_{\sD}(s)}p_{\sD}(s, d)P(s,a,d,s')V^{(k)}(s')$}}\label{step:minmax}
  \EndFor
  
  \EndWhile \label{step:endwhile}
  \Return Vector $\hat{V}$, where $\hat{V}(s) \leftarrow V^{(k)}(s)$
\State  Compute DIFT strategy $\hat{p}_{\sD}$,  $\hat{p}_{\sD}(s) \leftarrow \arg\max\limits_{p_{\sD} \in {\bf P}_{\sD}}\min\limits_{a \in \A_{\sA}(s)}\sum\limits_{s' \in {\bf S}}\sum\limits_{d \in \A_{\sD}(s)}p_{\sD}(s, d)P(s,a,d,s')\hat{V}(s')$
\end{algorithmic}
\end{algorithm}

\begin{lemma}\label{lem:value}
Consider the value iteration algorithm in Algorithm~\ref{alg:value}. Let $V^{(k+1)}$ and $V^{(k)}$ be the value vectors corresponding to iterations $k+1$ and $k$, respectively.  Then, $V^{(k+1)}(s) \geqslant V^{(k)}(s)$, for all $s \in {\bf S}$. 
\end{lemma}
\begin{proof}
We first prove that result for a state $s\in \{\phi, \tau_{\sA},\tau_{\sB} \}\cup \D$. From Eq.~\eqref{eq:val-iterate}, for every iteration $k=1,2,\ldots$ $V^{(k)}(s)=\beta$, for $s\in \{\phi, \tau_{\sA}\}$, and $V^{(k)}(s)=0$, for $s\in \{\tau_{\sB}\}\cup \D$. From the initialization step of Algorithm~\ref{alg:value} (Step~\ref{step:init}),   $V^{(0)}(s)=0$ for all $s \in {\bf S}$. Thus for any state $s$, where $s\in \{\phi, \tau_{\sA},\tau_{\sB} \}\cup \D$, $V^{(k+1)}(s) \geqslant V^{(k)}(s)$ for any arbitrary iteration $k$ of Algorithm~\ref{alg:value}.

For an arbitrary state $s$, where $s\in {\bf S} \setminus \{\phi, \tau_{\sA},\tau_{\sB} \}\cup \D$, we prove the result using an induction argument. The induction hypothesis is that arbitrary iterations $k$ and $k+1$ satisfy  $V^{(k+1)}(s) \geqslant V^{(k)}(s)$, for all $s\in {\bf S} \setminus \{\phi, \tau_{\sA},\tau_{\sB} \}\cup \D$. 

\noindent Base step: Consider $k=0$ as the base step. Initialize $V^{(0)}(s)=0$ for all $s\in {\bf S} \setminus \{\phi, \tau_{\sA},\tau_{\sB} \}\cup \D$ and  set $V^{(1)}(s) = 0$ for all $s\in {\bf S} \setminus \{\phi, \tau_{\sA},\tau_{\sB} \}\cup \D$. This gives $V^{(1)}(s) \geqslant V^{(0)}(s)$, for all $s\in {\bf S} \setminus \{\phi, \tau_{\sA},\tau_{\sB} \}\cup \D$.

\noindent Induction step: For the induction step, assume that iteration $k$ satisfies $V^{(k)}(s) \geqslant V^{(k-1)}(s)$ for all $s\in {\bf S} \setminus \{\phi, \tau_{\sA},\tau_{\sB} \}\cup \D$.

Consider iteration $(k+1)$. Then,
\begin{eqnarray}
\hspace*{-8 mm} V^{(k+1)}(s) \hspace*{-2.5 mm}&=& \hspace*{-4 mm} \max\limits_{p_{\sD} \in {\bf P}_{\sD}}\min\limits_{a \in \A_{\sA}(s)}\sum\limits_{s' \in {\bf S}}\sum\limits_{d \in \A_{\sD}(s)}\hspace*{-3 mm} p_{\sD}(s, d)P(s,a,d,s')V^{(k)}(s')\nonumber\\
\hspace*{-4.5 mm} &\geqslant &\hspace*{-4 mm} \min\limits_{a \in \A_{\sA}(s)}\sum\limits_{s' \in {\bf S}}\sum\limits_{d \in \A_{\sD}(s)}\hspace*{-3 mm} p^{(k)}_{\sD}(s, d)P(s,a,d,s')V^{(k)}(s')\label{eq:ind_2}\\
\hspace*{-4.5 mm} &\geqslant &\hspace*{-4 mm} \min\limits_{a \in \A_{\sA}(s)}\sum\limits_{s' \in {\bf S}}\sum\limits_{d \in \A_{\sD}(s)}\hspace*{-3 mm} p^{(k)}_{\sD}(s, d)P(s,a,d,s')V^{(k-1)}(s')\label{eq:ind_3}\\
\hspace*{-4.5 mm} &= &\hspace*{-4 mm}  \max\limits_{p_{\sD} \in {\bf P}_{\sD}}\min\limits_{a \in \A_{\sA}(s)}\sum\limits_{s' \in {\bf S}}\sum\limits_{d \in \A_{\sD}(s)}\hspace*{-3 mm} p_{\sD}(s, d)P(s,a,d,s')V^{(k-1)}(s')\nonumber\\
\hspace*{-0.1 mm} &=&  \hspace*{-2 mm}V^{(k)}(s)\label{eq:ind_4}
  \end{eqnarray}

Eq.~\eqref{eq:ind_2} holds as the value obtained using a maximizing policy is at least as the value obtained using a specific policy $p_{\sD}$. Eq.~\eqref{eq:ind_3} follows from the induction argument and Eq.~\eqref{eq:ind_4} is from the definition of  $V^{(k)}(s)$. This completes the proof.
\end{proof}

\begin{prop}[Monotone Convergence Theorem, \cite{real-analysis}]\label{prop:real}
 If a sequence is monotone increasing and bounded from above,
then it is a convergent sequence.
\end{prop}

The value of any state $s \in {\bf S}$ is bounded above by $\beta$. Using  Lemma~\ref{lem:value} and Proposition~\ref{prop:real}, we  state the  convergence result of Algorithm~\ref{alg:value}.

\begin{theorem}\label{th:value-conv}
Consider the value iteration algorithm in Algorithm~\ref{alg:value}. Let $V^{(k)}(s), V^{\*}(s)$ be the value at iteration $k$ and the optimal value of state $s \in {\bf S}$, respectively. Then, as $k\rightarrow \infty$, $V^{(k)}(s) \rightarrow V^{\*}(s)$, for all $s \in {\bf S}$. Further, the output of Algorithm~\ref{alg:value},  $\hat{p}_{\sD}$,  for $\delta \rightarrow 0$, is an optimal defender policy.
\end{theorem}

The value of any state $s \in {\bf S}$ is bounded above by $\beta$. From Lemma~\ref{lem:value} we know that the sequence $V^{(k)}(s)$ is a monotonically increasing sequence. By the monotone convergence theorem \cite{real-analysis}, a bounded and monotone sequence converges to the supremum, i.e., $\lim_{k \rightarrow \infty} V^{(k)}(s) \rightarrow V^{\*}(s)$, for all $s \in {\bf S}$. Thus the value iteration algorithm converges and the proof follows.
%
%
%

\begin{theorem}\label{th:exi}
For any $\epsilon >0$, Algorithm~\ref{alg:value} returns an $\epsilon$-Nash equilibrium of the APT-DIFT game.
\end{theorem}

Proof of Theorem~\ref{th:exi} follows from Lemma~\ref{lem:value}, Theorem~\ref{th:value-conv}, and Definition~\ref{def:epsilon-NE}.

The value of $\delta$ trades off  the desired  accuracy of the value vector and  the computational time. For a given value of $\delta$, the number of iterations of Algorithm~\ref{alg:value} depends on the size and connectivity structure of the IFG.  As the size and connectivity structure of the IFG varies across different datasets, we base the selection of $\delta$ solely on the accuracy desired by the security expert. In our experiments, we set $\delta=0.01$. Theorem~\ref{th:complexity} presents the computational complexity  of Algorithm~\ref{alg:value} for each iteration. 
\begin{theorem}\label{th:complexity}
Consider the APT-DIFT game  with $N$  number of nodes in the IFG and $q$  number of destination nodes. Let $\A_{\sA}$ and $\A_{\sD}$ be the action sets of  APT and  DIFT, respectively. Every iteration of Algorithm~\ref{alg:value}, i.e., steps~\ref{step:del}-\ref{step:endwhile},  has computational complexity $O({(N-q+1)}^2|\A_{\sA}||\A_{\sD}|)$.
\end{theorem}

\begin{proof}
Every iteration of Algorithm~\ref{alg:value} involves computation of the value vector. This involves solving, for every state $s \in {\bf S} \setminus \{\{\phi, \tau_{\sA}, \tau_{\sB} \} \cup \D\} $,  a linear program  of the form: 

\noindent {\it Maximize} $V(s)$\\
\begin{example}
\item[(1)]~$\sum\limits_{d \in \A_{\sD}(s)} p_{\sD}(s,d)=1$
\item[(2)]~$p_{\sD}(s,d) \geqslant 0$, for all $d \in \A_{\sD}(s)$
\item[(3)]~$V(s) \leqslant \sum\limits_{d \in \A_{\sD}(s)} Q(s,d,a)p_{\sD}(s,d)$, for all $a \in \A_{\sA}(s)$.
\end{example}

We note that $|{\bf S} \setminus \{\{\phi, \tau_{\sA}, \tau_{\sB} \} \cup \D\}|=N-q+1$. The above linear program has complexity of $O({(N-q+1)}|\A_{\sA}||\A_{\sD}|)$ \cite{boutilier1999decision}. Thus solving for all  $(N-q+1)$ states has a total complexity of $O({(N-q+1)}^2|\A_{\sA}||\A_{\sD}|)$.
\end{proof}

Algorithm~\ref{alg:value} is guaranteed to converge to an NE strategy asymptotically. When the state space of the APT-DIFT game has cycles, Algorithm~\ref{alg:value} updates the value vector  recursively    and consequently finite time convergence can not be guaranteed.   However, when the IFG is acyclic, the state space of the APT-DIFT game is acyclic (Theorem~\ref{th:finite}) and  a finite time convergence can be achieved using value iteration. 

Henceforth, the following assumption holds.

\begin{assume}\label{asm:IFG}
The IFG $\G$ is acyclic.
\end{assume}

The IFG obtained from the system log may not be acyclic in general. However, when the IFG is a cyclic graph, one can obtain an acyclic representation of the IFG   using the node versioning technique proposed in \cite{hossain2018dependence}. Throughout this subsection, we consider directed acyclic IFGs rendered using the approach in \cite{hossain2018dependence} and hence Assumption~\ref{asm:IFG} is non-restrictive.

 Theorem below  presents a termination  condition of the APT-DIFT game when the IFG is acyclic.
\begin{theorem}\label{th:finite}
Consider the  APT-DIFT game. Let Assumption~\ref{asm:IFG} holds and $N$ be the number of nodes in the IFG. Then the state space of the APT-DIFT game is acyclic and the game terminates in at most $N+4$  number of steps. 
\end{theorem} 
\begin{proof}
Consider any arbitrary strategy pair $(p_{\sA}, p_{\sD})$. We first prove the acyclic property of the state space of the game under $(p_{\sA}, p_{\sD})$.  The state space $\bS$ is constructed by augmenting the IFG with states $v_0, \phi$, $\tau_{\sA}$, and $\tau_{\sB}$.
We note that a state $s \in \{ \phi, \tau_{\sA}, \tau_{\sB}\} \cup \D$ does not lie in a cycle in $\bS$ as $s$ is an absorbing state and hence have no outgoing edge, i.e., $\A_{\sA}(s)=\A_{\sD}(s)=\emptyset$. The state $v_0$ does not lie in a cycle in $\bS$ as there are no incoming edges to $v_0$.  This concludes that a state $s \in \{v_0, \phi, \tau_{\sA}, \tau_{\sB}\}\cup \D$ is not part of a cycle in $\bS$. Thus a cycle can possibly exist in $\bS$ only if there is a cycle which has some states in  $v_{q+1}, \ldots, v_{N}$, since $\D = \{v_{1}, v_{2}, \ldots, v_q \}$. Recall that states $v_1, \ldots, v_{N}$ correspond to nodes of $\G$. As $\G$ is acyclic, there are no cycles in $\bS$. Since $\bS$ is acyclic under any arbitrary strategy pair $(p_{\sA}, p_{\sD})$ and the state space has finite cardinality, the game terminates in finite number of steps. Further, since $|\bS|=N+4$, $T \leqslant N+4$.
\end{proof}

Using Theorem~\ref{th:finite}, we propose a value iteration algorithm with guaranteed finite time convergence. We will use the following definition in our approach. 

  For a directed acyclic graph, topological ordering of the node set is defined below.
\begin{defn}[\cite{kahn1962topological}]
 A  topological ordering of a directed graph is a linear ordering of its vertices such that for every directed edge $(u, v)$ from vertex $u$ to vertex $v$, $u$ comes before $v$ in the ordering.
\end{defn}
For a directed acyclic graph with vertex set $B$ and edge set $E$ there exists an algorithm of complexity $O(|B| + |E|)$ to find the topological order \cite{kahn1962topological}. Using the topological ordering, one can find a  hierarchical level partitioning of the nodes of a directed acyclic graph.
   Let $\S$ be the topologically ordered set of nodes of $\bS$.  
 Let the number of hierarchical levels associated with $\S$ be $M$, say $L_1, L_2, \ldots, L_M$. Then $L_1=v_0$, $L_M = \{\phi, \tau_{\sA}, \tau_{\sB} \} \cup \D$. Moreover, 
 a state $s$ is said to be in hierarchical level $L_j$ if there exists an edge into $s$ from a state $s'$ which is in some level  $L_{j'}$, where $j'<j$, and there is no edge into $s$ from  a state $s''$ which is in level $L_{j''}$, where $j'' > j$.
 
 Algorithm~\ref{alg:value-2} presents the pseudocode to solve the APT-DIFT game using the topological ordering and the hierarchical levels, when the IFG is acyclic.  
 
\begin{algorithm}[t]
  \caption{Value iteration algorithm to find NE strategy for APT-DIFT game under Assumption~\ref{asm:IFG}} \label{alg:value-2}
  \begin{algorithmic}
\State \textit {\bf Input:} APT-DIFT game with state space ${\bf S}$, destination set $\D$, action space $\A_{\sD}, \A_{\sA}$, payoff parameter $\beta$, transition probabilities $FN(\cdot), FP(\cdot)$
\State \textit{\bf Output:} Value vector  $V^{\*}$ and  defender strategy $p^{\*}_{\sD}$
\end{algorithmic}
  \begin{algorithmic}[1]
  \State Find the topological ordering of the state space graph, $\S$
\State Using $\S$, obtain the set of nodes corresponding to hierarchical levels $L_1, L_2, \ldots, L_M$
  \State Initialize value vector $V(s) \leftarrow 0$, for all $s \in {\bf S}$ and $V(s') \leftarrow \beta$ for $s' \in \{\tau_{\sA}, \phi\}$, $V(s'') \leftarrow 0$ for all $s'' \in{\bf S}\setminus \{\tau_{\sA}, \phi\}$\label{step:init2}
  \For{$k \in \{M-1, M-2, \ldots, 1 \}$}
  \For{$s \in L_k$}
  \State {\scalefont{0.98}{\hspace*{-3 mm}$V(s) \leftarrow \max\limits_{p_{\sD} \in {\bf P}_{\sD}}\min\limits_{a \in \A_{\sA}(s)}\sum\limits_{s' \in {\bf S}}\sum\limits_{d \in \A_{\sD}(s)}\hspace*{-2 mm}p_{\sD}(s, d)P(s,a,d,s')V(s')$}}\label{step:minmax2}
  \EndFor
   \State $k \leftarrow k+1$
  \EndFor
  \Return Value vector $V^{\*} \leftarrow V$
\State Compute DIFT strategy $p^{\*}_{\sD}$,  $p^{\*}_{\sD}(s) \leftarrow \arg\max\limits_{p_{\sD} \in {\bf P}_{\sD}}\min\limits_{a \in \A_{\sA}(s)}\sum\limits_{s' \in {\bf S}}\sum\limits_{d \in \A_{\sD}(s)}p_{\sD}(s, d)P(s,a,d,s')V^{\*}(s')$
\end{algorithmic}
\end{algorithm}

\begin{cor}
Consider the APT-DIFT game and let $\A_{\sA}, \A_{\sD}$ be the action sets of  APT,  DIFT, respectively. Let the IFG is acyclic with  $N$  number of nodes and $q$  number of destination nodes.  Then, Algorithm~\ref{alg:value-2} returns the value vector $V^{\*}$. Moreover, Algorithm~\ref{alg:value-2} has computational complexity $O({(N-q+1)}^2|\A_{\sA}||\A_{\sD}|)$.
\end{cor}
\begin{proof}
Recall that under Assumption~\ref{asm:IFG}, the state space ${\bf S}$ of the APT-DIFT game is acyclic. Thus the topological ordering $\S$ and the hierarchical levels, $L_1, \ldots, L_M$,  of ${\bf S}$ can be obtained in polynomial time \cite{kahn1962topological}.  We note that, values at the absorbing states, i.e., states at hierarchical level $M$, are   $V(s) = \beta$ for $s \in \{\tau_{\sA}, \phi\}$ and $V(s') = 0$ for all $s' \in \{\tau_{\sB}\} \cup \D$.   Using the hierarchical levels, the value vector can be computed in a dynamic  programming manner starting from states in the last hierarchical level. Initially, using the values of the states at level $M$, the values of states at level $M-1$ can be obtained. Similarly, using values of states at level $M-1$ and level $M$, the values of states at level $M-2$ can be obtained and so on. By  recursive computations we obtain the value vector $V^{\*}$ and the corresponding  DIFT strategy $p^{\*}_{\sD}$.

  Finding the topological ordering and the corresponding hierarchical levels of  the state space graph has $O(|{\bf S}|^2)$ computations. The linear program associated with each state in the value iteration involves $O({(N-q+1)}|\A_{\sA}||\A_{\sD}|)$ computations (See proof of Theorem~\ref{th:complexity}). Since we need to compute values corresponding to $(N-q+1)$ states,  complexity of Algorithm~\ref{alg:value-2} is $O({(N-q+1)^2}|\A_{\sA}||\A_{\sD}|)$.
\end{proof}

\subsection{Hierarchical Supervised Learning (HSL) Algorithm}
This subsection  presents our approach to solve the APT-DIFT game when the game model is not fully  known. It is often unrealistic to know the precise values of the false-positive rates and the false-negative rates of DIFT as these values are  estimated empirically. More specifically, the transition probabilities of the APT-DIFT game are unknown.  The traditional reinforcement learning algorithms (Q-learning)  for constant-sum games with unknown transition probabilities \cite{hu1998multiagent} assume perfect information, i.e., the players can observe the past actions and rewards of the opponents  \cite{hu1998multiagent}. 

On the other hand, dynamic programming-based  approaches, including value iteration and policy iteration algorithms, require the transition probabilities in the computations.  When $FN(\cdot)$ and $FP(\cdot)$ values are unknown,  the optimization problem associated with the states in step~\ref{step:minmax} of  Algorithm~\ref{alg:value} and step~\ref{step:minmax2} of Algorithm~\ref{alg:value-2} can not be solved. That is, in the LP
 
\begin{prob}\label{prob:LP}
\noindent {\it Maximize} $V(s)$\\
\begin{example}
\item[(1)]~$\sum\limits_{d \in \A_{\sD}(s)} p_{\sD}(s,d)=1$
\item[(2)]~$p_{\sD}(s,d) \geqslant 0$, for all $d \in \A_{\sD}(s)$
\item[(3)]~$V(s) \leqslant \sum\limits_{d \in \A_{\sD}(s)} Q(s,d,a)p_{\sD}(s,d)$, for all $a \in \A_{\sA}(s)$,
\end{example}
\end{prob}
the $Q(s,a,d)$ values are unknown, where

 $$Q(s,a,d) = \sum_{s' \in {\bf S}}P(s,a,d,s')V(s'),$$ and $P(s,a,d,s')$ is defined  in Eq.~\eqref{eq:transition_D}. To the best of our knowledge, there are no known algorithms, with guaranteed equilibrium convergence, to solve imperfect and incomplete stochastic game models as that of the APT-DIFT game presented in this paper.
 To this end, we propose a supervised learning-based approach to solve the APT-DIFT game.  Our approach consists of two key steps.
 
 \begin{enumerate}
\item[(1)] Training a neural network to approximate the value vector of the game for a given strategy pair, and  \item[(2)]  A policy iteration algorithm to compute an $\epsilon$-optimal NE strategy.
   \end{enumerate}
   
%
%
 
 Our approach to solve the APT-DIFT game, when the rates of false-positives and false-negatives are unknown, utilizes the topological ordering and the hierarchical levels of the state space graph. 
  We propose a hierarchical supervised learning (HSL)-based approach, that predicts the $Q$-values of the APT-DIFT game and then solve the LP for all the states (Problem~\ref{prob:LP}) in a hierarchical manner, to approximate an NE strategy  of the APT-DIFT game. 
  In HSL, we utilize a neural network to obtain a mapping between the strategies of the players and the value vector. This mapping is an approximation, using which HSL presents an approach to compute an approximate equilibrium and evaluates the performance empirically using real attack datasets.
  
  In order to predict the $Q$-values of the game, we train a neural network to learn the value vector of the game and use the trained model to predict the $Q$-values. The reasons to train an NN to learn the value vector instead of directly learning the $Q$-values are: (a)~while the value vector is of dimension $|\bf S|$ the $Q$-values have dimension $|{\bf S}||\A_{\sA}||\A_{\sD}|$ and (b)~generation of data samples for value vector is easy. The approach for data generation and training is elaborated below.

 We note that, it is possible to simulate the APT-DIFT game and observe the final game outcome, i.e., the payoffs of the players.  For a given $(p_{\sA}, p_{\sD})$,  the value at state $s$, $V(s)$, is the payoff of the defender if the game originate at state $s$, i.e.,  $V(s) = U_{\sD}(s, p_{\sA}, p_{\sD})$.   Training the neural network for predicting the value vector of the APT-DIFT game consists of two steps.
\begin{enumerate}
\item[(i)] Generate random samples of strategy pairs, $(p_{\sA}, p_{\sD})$ 
\item[(ii)] Simulate the APT-DIFT game for each of the randomly generated sample of strategy pair and obtain the values corresponding to all states.
  \end{enumerate}  
%
%

The neural network takes as input the strategy pairs and outputs the value vector. The training is done using a multi-input, multi-output neural network, represented as $\F: X \rightarrow Y$, where $X 
\subseteq [0,1]^{|\A_{\sA}|} \times [0,1]^{|\A_{\sD}|}$ and $Y \subseteq \mathbb{R}^{|{\bf S}|}$. The neural network may not compute the exact value vector, however, it can approximate the value vector to arbitrary  accuracy.
  Given a function $f(x)$ and $\xi >0$, the guarantee is that by
 using enough hidden neurons it is  always possible to find a neural network whose output $g(x)$ 
  satisfies $|f(x)-g(x)|< \xi $, for all inputs $x$ \cite{saji}. In other words, the approximation will be good to within the desired accuracy for every possible input. However, the training method does not guarantee that the neural network obtained at the end of the training process is one that satisfies the specified level of accuracy. 
   Using the trained neural network we predict the $Q$-values of the game.
  \begin{lemma}\label{lem:NN}
  Consider the APT-DIFT game with state space ${\bf S}$. Let a neural network be trained using samples of strategy pairs $(p_{\sA}, p_{\sD})$ to predict the value vector $V$ of the game such that the mean absolute error is within the desired tolerance of $\zeta \leqslant 0.01$. Then, the trained neural network also yield the $Q$-values, $Q(s,a,d)$, for all $s \in {\bf S}$, $a \in \A_{\sA}$, and $d \in \A_{\sD}$.
  \end{lemma}
    \begin{proof}
Consider a neural network that is trained using enough samples of strategy pairs  to predict the value vector  of the game. Thus given a strategy pair $(p_{\sA}, p_{\sD})$, the neural network predicts $V=\{V(s):s \in {\bf S} \}$.  Consider an arbitrary state $s$. 
The $Q$-value corresponding to state $s$ and action pair $a,d$ for the APT and DIFT, respectively, is given by $Q(s,a,d) = \sum_{s' \in {\bf S}}P(s,a,d,s')V(s')$. 
For a  strategy $(p_{\sA}, p_{\sD})$ with $p_{\sA}(s)$ such that $p_{\sA}(s,a)=1$ and $p_{\sD}(s)$ such that $p_{\sD}(s,d)=1$ gives
\begin{eqnarray*}
V(s) &=& \sum_{s' \in {\bf S}}\sum_{a\in \A_{\sA}}\sum_{d \in \A_{\sD}}p_{\sA}(s,a)p_{\sD}(s,d)P(s,a,d,s')V(s')\\
&=& \sum_{s' \in {\bf S}}P(s,a,d,s')V(s')\\
&=& Q(s,a,d).
\end{eqnarray*}  
Hence  the $Q$-values of the game can be obtained using the neural network that is trained to predict the value vector by inputting a strategy pair $(p_{\sA}, p_{\sD})$ with $p_{\sA}(s,a)=1$ and $p_{\sD}(s,d)=1$, for any  $s \in {\bf S}$, $a \in \A_{\sA}$, and $d \in \A_{\sD}$.
  \end{proof}

Using the trained neural network we run a policy iteration algorithm.  In the policy iteration, we update the strategies of both players, APT and DIFT, by solving the stage games in a dynamic programming manner. The details of the algorithm are presented below.

\begin{algorithm}[h]
  \caption{HSL  Algorithm for APT-DIFT game 
  \label{alg:twostage}}
  \begin{algorithmic}
\State \textit {\bf Input:} APT-DIFT game with state space ${\bf S}$, destination set $\D$,  action sets $\A_{\sA}, \A_{\sD}$, payoff parameter $\beta$
\State \textit{\bf Output:} Value vector $\hat{V}$ and defender strategy $\hat{p}_{\sD}$
\end{algorithmic}
\begin{algorithmic}[1]
\State  Generate random samples of $(p_{\sA}, p_{\sD})$ and value vector\label{step:samples}
\State Train $\F$ using the data set from Step~\ref{step:samples}\label{step:train}

\State Find the topological ordering of the state space graph, $\S$\label{step:top}
\State Obtain the set of nodes corresponding to hierarchical levels $L_1, L_2, \ldots, L_M$\label{step:level}
\State Initialize $V(\tau_{\sA})\leftarrow\beta$, $V(\phi)\leftarrow\beta$, $V(\tau_{\sB})\leftarrow 0$, and $V(s)\leftarrow 0$ for $s \in \D$\label{step:ini-1}
\State Initialize randomly $(\hat{p}^{(M-1)}_{\sA}, \hat{p}^{(M-1)}_{\sD})$\label{step:ini-2}

\For {$k \in \{M-1, M-2, \ldots, 1\} $}
\For{$s\in L_k$}\label{step:s}
\For{$ a \in \A_{\sA}(s), d \in \A_{\sD}(s)$} 
 \State Update $(\hat{p}^{(k)}_{\sA}, \hat{p}^{(k)}_{\sD})$ such that $\hat{p}^{(k)}_{\sA}(s,a)=1$ and $\hat{p}^{(k)}_{\sD}(s,d)=1$
 \State Predict the value vector $\hat{V}$ using the neural network $\F$, $\hat{V}\leftarrow \F(\hat{p}^{(k)}_{\sA}, \hat{p}^{(k)}_{\sD})$
 \State Assign $\hat{Q}(s,a,d) \leftarrow \hat{V}(s)$\label{step:Q}
 \EndFor
 \State Assign $Q(s,a,d) \leftarrow \hat{Q}(s,a,d)$ and solve Problem~\ref{prob:LP} for  $s$ to obtain $V(s)$ and $p_{\sD}(s)$ 
 \State Update $ \hat{p}^{(k)}_{\sD}$ such that $ \hat{p}^{(k)}_{\sD} (s) \leftarrow  p_{\sD}(s) $
  \State Find the action  corresponding to the minimum entry of the vector $\hat{Q}(s,a,d)\,\hat{p}^{(k)}_{\sD} (s, d)$, say $\hat{a}$
   \State Update $ \hat{p}^{(k)}_{\sA}$ such that $ \hat{p}^{(k)}_{\sA} (s, \hat{a}) \leftarrow  1$
 \EndFor
\EndFor
 \Return $\hat{p}_{\sD} \leftarrow \hat{p}^{(k)}_{\sD}$
\end{algorithmic}
\end{algorithm}

Algorithm~\ref{alg:twostage} presents the pseudocode for our HSL algorithm to solve the APT-DIFT game. Initially we generate random samples of strategy pairs $(p_{\sA}, p_{\sD})$ and value vector $V$ and train a neural network. Using the trained neural network, we propose a strategy iteration algorithm. As Assumption~\ref{asm:IFG} holds, the state space graph ${\bf S}$ is acyclic. Thus one can compute the topological ordering (Step~\ref{step:top}) and the hierarchical levels (Step~\ref{step:level}) of  ${\bf S}$ in polynomial time. The values at the absorbing states are known and are hence set as $V(\tau_{\sA})=V(\phi)=\beta$, $V(\tau_{\sB})=0$, and $V(s)=0$ for all $s \in \D$. As the values of the states at level $M$ are known, the algorithm begins with level $M-1$.  Initially the strategies of APT and DIFT are randomly set to $\hat{p}^{(M-1)}_{\sA}$ and $\hat{p}^{(M-1)}_{\sD}$, respectively. 
Then we compute the $Q$-values of all the states using the trained neural network, following the hierarchical order. Due to the hierarchical structure of the state space, computation of $Q$-value of a state in $j^{\rm th}$ hierarchical level depends only on the states that are at levels higher than $j$. Consider  level $M-1$ and an arbitrary state $s$ in level $M-1$. The $Q$-values of $s$ are predicted using the trained neural network by selecting suitable deterministic strategy pairs. That is, for predicting $Q(s,a,d)$, choose strategies such that $\hat{p}^{(M-1)}_{\sA}(s,a)=1$ and $\hat{p}^{(M-1)}_{\sD}(s,d)=1$ as input to the neural network.  Using the $Q$-values of state $s$, solve the LP and obtain value vector, DIFT strategy, and the optimal action of the APT. The strategies $\hat{p}^{(M-1)}_{\sA}$ and $\hat{p}^{(M-1)}_{\sD}$ are updated using the output of the LP such that the $\hat{p}^{(M-1)}_{\sA}(s)$ and $\hat{p}^{(M-1)}_{\sD}(s)$ corresponds to an NE strategy. Once the strategies of both players are updated for all the states in level $M-1$, the process continues for level $M-2$ and so on.

In  HSL algorithm, the prediction of $Q$-values using the neural network (Step~\ref{step:s}-\ref{step:Q}) corresponding to states in a particular level can be run parallel. To elaborate, let there are $x$ number of states in level $L_j$. Then select an action pair $(a,d)$ corresponding to each state and run the prediction step of the algorithm. Thus, every run of the neural network can predict $x$ number of $Q$-values, for $x$ different states.

{\bf Remark}: APT attacks typically consist of multiple stages, e.g., initial compromise, internal reconnaissance, foothold establishment, and data exfiltration, with each stage having a specific set of targets. To capture the multi-stage nature of the attack, we construct a multi-stage IFG, $\G_m$, from the IFG $\G$. Consider an attack that consists of $m$ attack stages with destinations of stage~$j$  denoted as $\D_j$. Then we duplicate $m$ copies of the IFG $\G$ such that nodes in $\D_j$  in the $j^{\rm th}$ copy of $\G$ is connected to respective nodes in $\D_j$ in the $(j+1)^{\rm th}$ copy, for $j \in \{1, \ldots, m\}$.   Also, set $\D_m = \D$. The construction of $\G_m$ guarantees that any path that originate in set $\lambda$ in the first copy and terminate in $\D$ in the $m^{\rm th}$ copy is a feasible attack path. Figure~\ref{fig:multiIFG} shows schematic diagram of a multi-stage IFG.
 For notational brevity the paper presents the single-stage attack case, i.e.,  $m=1$.  All the algorithm and results presented in this paper also apply to the case of multi-stage attack using $\G_m$ as the IFG.
\begin{figure}[t]
	\includegraphics[width=0.48\textwidth]{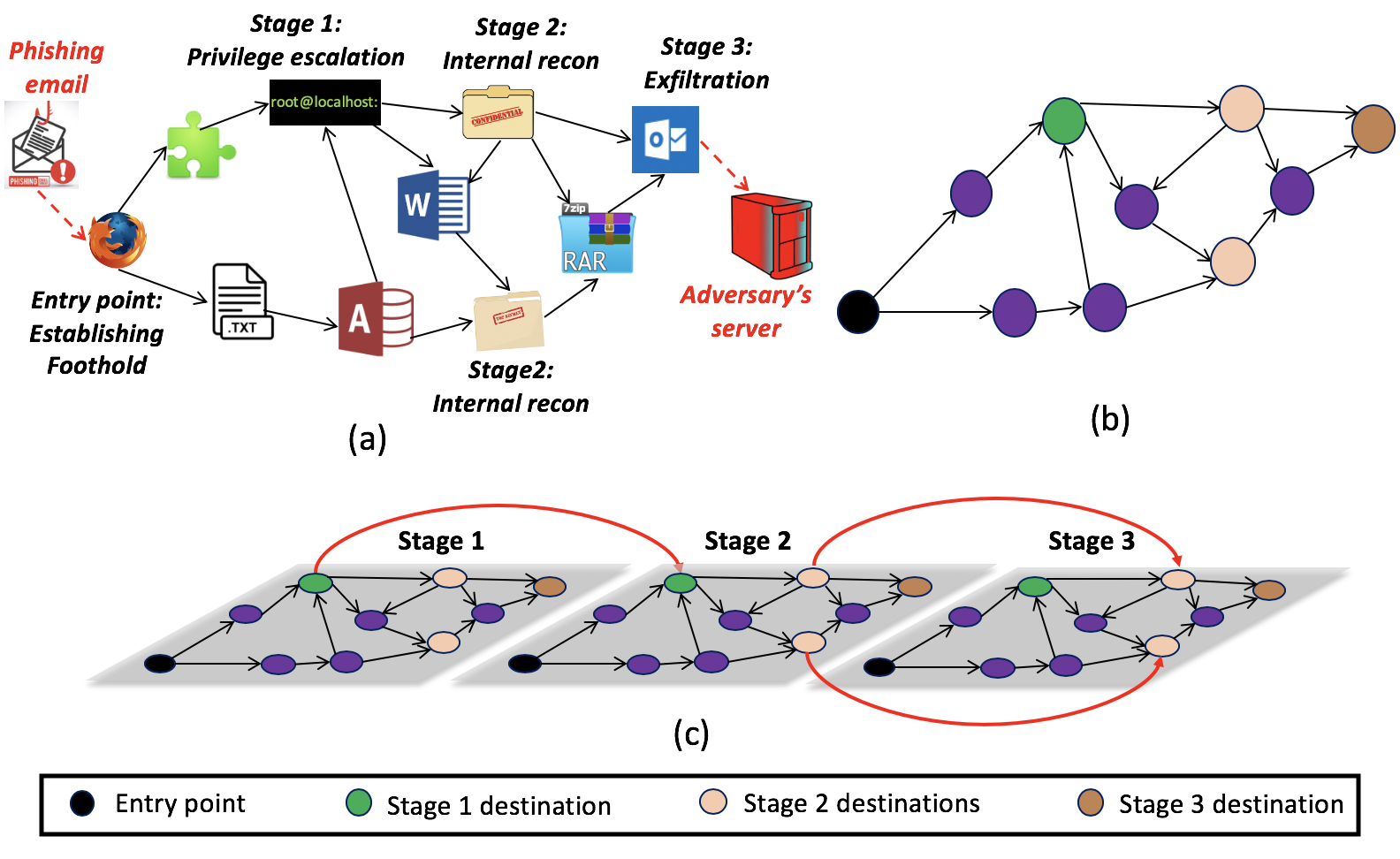}
	\caption{\small An example of a multi-stage APT attack scenario (Figure~\ref{fig:multiIFG}(a)). An illustrative diagram of the IFG $\G$ (Figure~\ref{fig:multiIFG}(b)) and the corresponding multi-stage IFG $\G_m$ (Figure~\ref{fig:multiIFG}(b)) that consists of three stages of the attack, i.e., $m=3$. The propagation of the attack from  stage~$j$ of the attack to stage~$(j+1)$ is captured in $\G_m$ by connecting the destination nodes ${\D}_j$ in stage~$j$ to their respective nodes in stage~$(j+1)$, for $j =1,2$.}\label{fig:multiIFG}
\end{figure}

\section{Simulation}\label{sec:sim}

In this section we test and validate Algorithms~\ref{alg:value},~\ref{alg:value-2}~and~HSL algorithm (Algorithm~\ref{alg:twostage})  on two real world attack datasets. First we provide the details on the attack datasets and explain the construction of IFGs corresponding to each attack from their respective system log data. Then we discuss our experiments and present the results.

\subsection{Attack Datasets}

We use two attack datasets in our experiments. First attack dataset is corresponding to a nation state attack \cite{brenner2009cyberthreats} and the second is related to a ransomware attack \cite{mohurle2017brief}. Each attack was executed individually in a computer running {\em{Linux}} operating system and system logs were recorded through RAIN system \cite{ji2017rain}. System logs contain records related to both malicious and benign information flows.



\subsubsection{Nation State Attack}
Nation state attack (a state-of-the-art APT attack) is a three day adversarial engagement orchestrated by US DARPA red-team during the evaluation of RAIN system. Attack campaign was designed to steal sensitive proprietary and personal information from the victim system. In our experiments we used the system logs recorded during the day 1 of the attack. The attack consists of four key stages: initial compromise, internal reconnaissance,  foothold establishment, and data exfiltration. Our experiments considered the first stage of the attack, i.e., initial compromise stage, where APT used spear-phishing to lead the victim user to a website that was hosting ads from a malicious web server. After navigating to the website, APT exploited a vulnerability in the Firefox browser and compromised the Firefox. 



\subsubsection{Ransomware Attack}
Ransomware attack campaign was designed to block access to the files in $./home$ directory and demand a payment from the victim in exchange for regranting the access. The attack consists of three stages: privilege escalation, lateral movement of the attack, and encrypting and deleting $./home$ directory. Upon reaching the final stage of the attack, APT first opened and read all the files in the $./home$ directory of the victim computer. Then APT wrote the content of the files in $./home$  into an encrypted file named $ransomware.encrypted$. Finally, APT  deleted the $./home$ directory.  

\subsection{Construction of IFG}\label{subsec:prune}
Direct conversion of the system logs into IFG typically results in coarse graphs with large number of nodes and edges as it includes all the information flows recorded during the system execution. Majority of these information flows are related to the system's background processes (noise) which are unrelated to the actual attack campaign and it is computationally intensive to run the proposed algorithms on a state space induced by such a coarse graph. Hence, we use the following pruning steps to prune the coarse graph without losing any attack related causal information flow dependencies. 
\begin{enumerate}
	\item When multiple edges with same directed orientation exist between two nodes in the coarse IFG, combine them to a single directed edge.  For example, consider a scenario where multiple ``read" system calls are recorded between a file and a process.
		This results in multiple edges between the two nodes of the resulting coarse IFG. Our APT-DIFT game formulation only requires to realize the feasibility of transferring information flows between pairs of processes and files. Hence, in scenarios similar to the above  example,  we collapse all the multiple edges between the two nodes in the coarse IFG into a single edge.
	\item Find all the nodes in  coarse IFG  that have at least one information flow path from an entry point of the attack to a target of the attack. When attack consists of multiple stages find all the nodes in coarse IFG that have at least one information flow path from a destination of stage $j$ to a  destination of a stage $j+1$, for all $j \in \{1, \ldots, m-1\}$. 
	\item From coarse graph, extract the subgraph corresponding to the entry points, destinations, and the set of nodes found in Step $2$. 
	\item Group all the nodes that correspond to files of a single directory to a single node related to the  parent file directory. For example, assume the resulting coarse IFG has three files, $./home/inventory/prices.xlsx$, $./home/vendors/contacts/addresses.doc$, and $./home/vendors/ledger.db$,  that are uniquely identified by their respective file paths. In this case we will group all three files, $prices.xlsx$, $addresses.doc$, and $ledger.db$ under one super-node corresponding to the parent directory $./home$.  Incoming and outgoing edges associated with each of the three files are then connected to the new node $./home$. If any subset of the new edges connected to $./home$ induce multiple edges with same orientation to another node in the IFG, then follow step~1). It is also possible to group the files into respective sub-directories such as $./home/inventory/$ and $./home/vendors/$ given in the example. Such groupings will facilitate much finer-grained security analysis at the cost of larger IFGs which require more computation resources to run the proposed APT-DIFT game algorithms presented in this paper.
	\item If the resulting graph after Steps $1$-$4$ contains any cycles use \emph{node versioning} techniques \cite{hossain2018dependence} to remove cycles while preserving the information flow dependencies in the graph.
\end{enumerate}
Steps $2$ and $3$ are done using upstream, downstream, and {\color{black}point-to-point} stream pruning techniques mentioned in \cite{ji2017rain}. The resulting  information flow graph is called  pruned IFG. We tested the proposed algorithms on the state spaces of APT-DIFT games corresponding to these pruned IFGs. 

For the nation state attack, initial conversion of the system logs into an IFG resulted in a coarse graph with 299 nodes and 404 edges which is presented in Figure~\ref{fig:IFG_nationstate}(a).. We used steps $1$ to $4$ explained in Section~\ref{subsec:prune} to obtain a pruned IFG with $30$ nodes and $74$ edges. A network socket connected to an untrusted IP address was identified as an entry point of the attack and the target of the attack is Firefox process. Figure~\ref{fig:IFG_nationstate}(b) shows the pruned IFG of nation state attack. In the IFG , there are 21 information flow paths from the entry point to the target Firefox.

\begin{figure}[t]
	\centering
	{\includegraphics[width=0.475\textwidth]{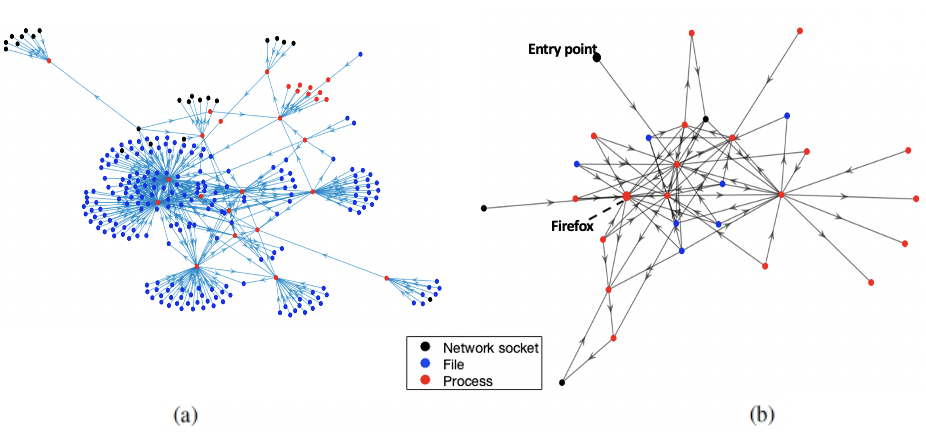}}
	\hspace{\parindent}
	\vspace{-5mm}
	\caption{\small {Relevant attack information of nation state attack: (a)~coarse IFG and (b)~pruned IFG. Nodes of the graph are color coded to illustrate their respective types (network socket, file, and process). A network socket is identified as the entry points of the nation state attack. Target of the attack is {\it Firefox} process.}}\label{fig:IFG_nationstate}
\end{figure}
\begin{figure}[t]
	\centering
	\includegraphics[width=8.5cm]{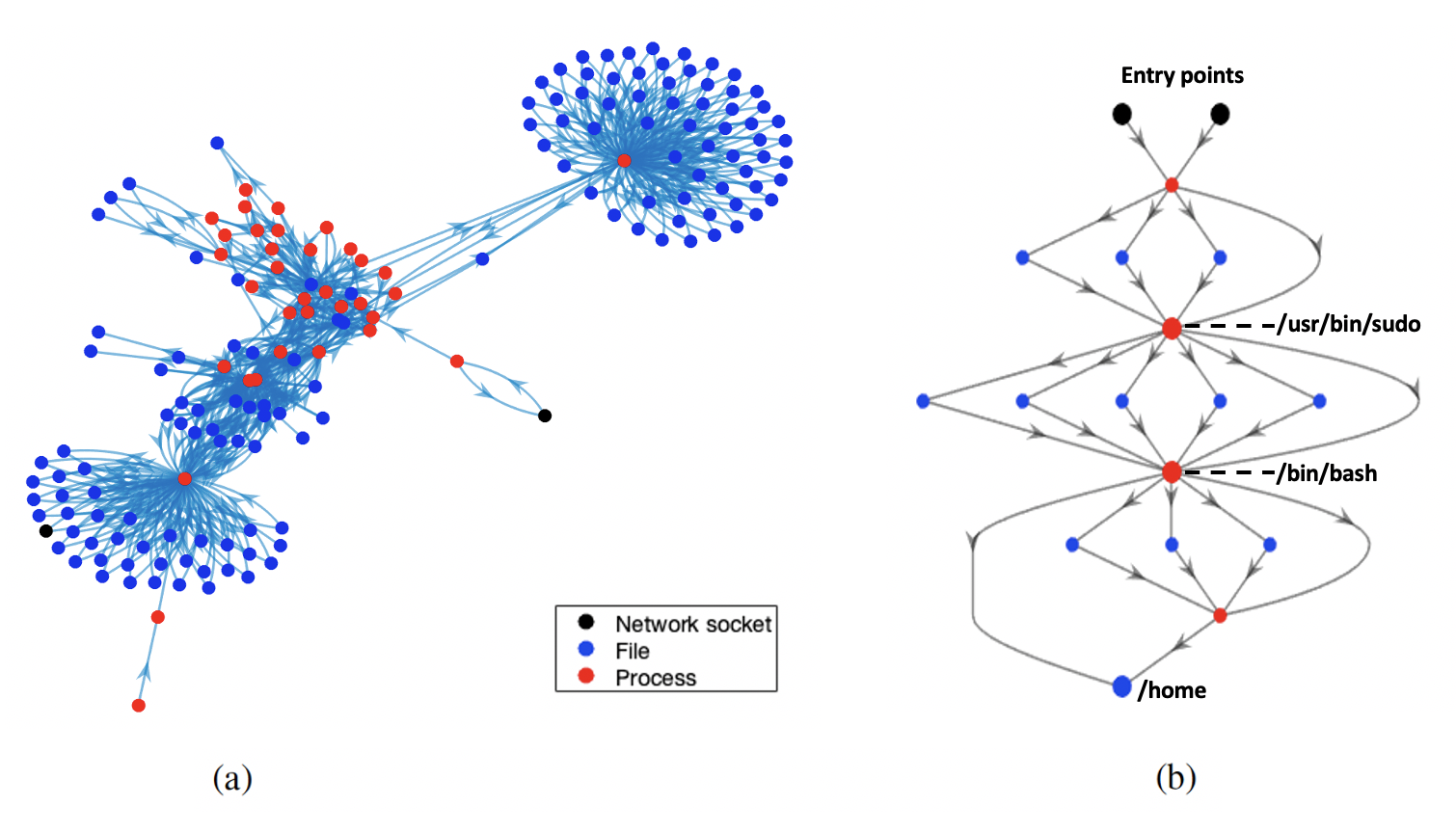}
		\vspace{-3mm}
	\caption{ \small Relevant attack information of ransomware attack: (a)~coarse IFG and (b)~pruned IFG. Nodes of the graph are color coded to illustrate their respective types (network socket, file, and process). Two network sockets are identified as the entry points of the ransomware attack. Targets of the attack ($/usr/bin/sudo, /bin/bash, /home$) are  labeled in the graph.}\label{fig:IFG_ransomware}
			\vspace{-3mm}
\end{figure}
\begin{figure*}[t]
	\centering
		\begin{subfigure}[b]{0.41\textwidth}
	{\includegraphics[width=0.7\textwidth]{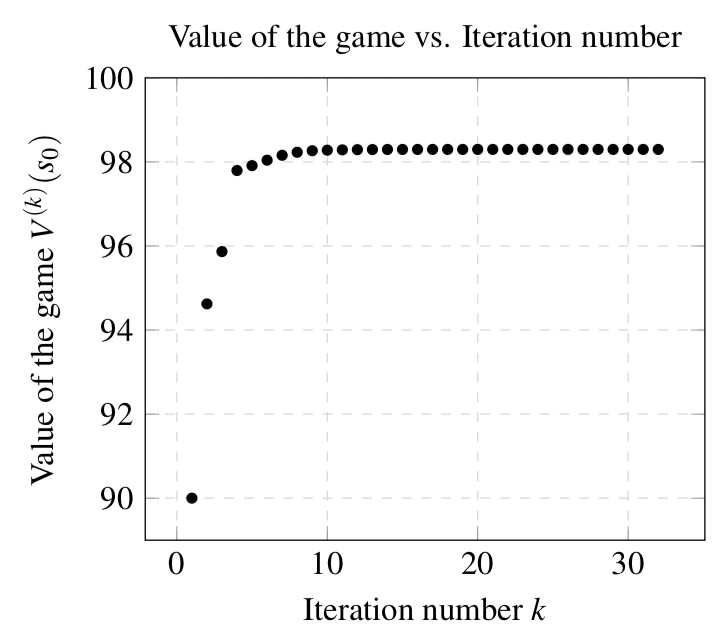}}
		\caption{}\label{fig:value}
	\end{subfigure}~\hspace{1 cm}
	\begin{subfigure}[b]{0.41\textwidth}
		\centering
	{\includegraphics[width=0.7\textwidth]{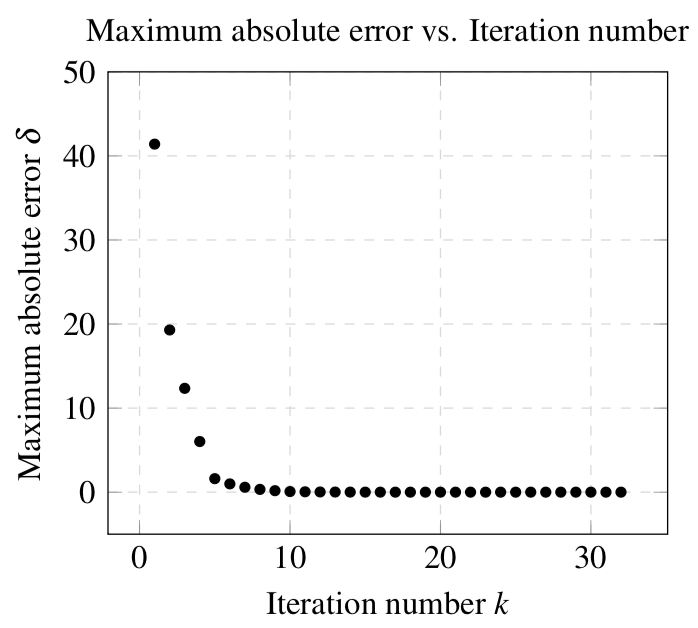}}
		\caption{}\label{fig:error}
	\end{subfigure}~\hspace{0.2 cm}
	\vspace*{-3 mm}
	\caption{\small Figure~\ref{fig:Algo_one}(a) plots the value of APT-DIFT game $V^{(k)}(s_0)$ corresponding to pruned IFG of nation state attack given in Figure~\ref{fig:IFG_nationstate},  computed at iterations $k = 1, 2, \ldots, 32$ in Algorithm~\ref{alg:value}. Figure~\ref{fig:Algo_one}(b) plots the maximum absolute error $\delta= \max_{s \in {\bf S}} |V^{(k)}(s) - V^{(k-1)}(s)|$, for $k = 1, 2, \ldots, 32$. The payoff parameter is set to $\beta = 100$.}\label{fig:Algo_one}
\end{figure*}
\begin{figure*}[t]
	\centering
	\begin{subfigure}[b]{0.41\textwidth}
		\includegraphics[width=8.5cm]{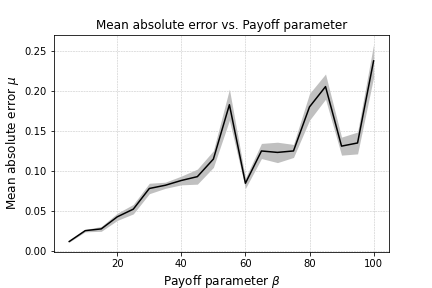}
		\caption{}\label{fig:beta}
	\end{subfigure}~\hspace{1 cm}
\begin{subfigure}[b]{0.41\textwidth}
	\includegraphics[width=8.5cm]{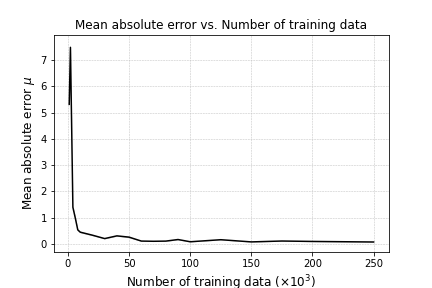}
	\caption{}\label{fig:TD_len}
\end{subfigure}~\hspace{1 cm}
\vspace*{-3 mm}
	\caption{\small Figure~\ref{fig:Algo_Two}(a) shows the mean absolute error $\mu$ between actual value vector $V$ and estimated value vector $\hat{V}$ (using neural network) for different payoff parameter values $\beta = 5, 10, \ldots, 100$. We used $10^{5}$ training data samples in the HSL algorithm (Algorithm~\ref{alg:twostage}). 
		Figure~\ref{fig:Algo_Two}(b) shows variation in $\mu$ with respect to the number of training data samples used in HSL algorithm. In this experiment $\beta = 50$. Each data point in both cases are calculated by averaging results over $10$ independent trials. Bars in each data point shows the standard errors associated with the $\mu$ values obtained in the different trials. Pruned IFG corresponding to the ransomware attack in Figure~\ref{fig:IFG_ransomware} was used in both cases.}\label{fig:Algo_Two}
\end{figure*}

For the ransomware attack, direct conversion of the system logs resulted in a coarse IFG with $173$ nodes and $482$ edges which is presented in Figure~\ref{fig:IFG_ransomware}(a). We pruned the resulting coarse IFG using the steps given in Section~\ref{subsec:prune}. The pruned IFG of ransomware attack consists of $18$ nodes and $29$ edges. Two network sockets that indicate series of communications with external IP addresses in the recorded system logs were identified as the entry points of the attack. Figure~\ref{fig:IFG_ransomware}(b) illustrates the pruned IFG of ransomware attack with annotated targets $/usr/bin/sudo, /bin/bash, /home$.
%


\subsection{Experiments and Results}
Algorithm~\ref{alg:value} was implemented on APT-DIFT game model associated with the pruned cyclic IFG  of nation state attack given in Figure~\ref{fig:IFG_nationstate}.  Figure~\ref{fig:Algo_one}(a) shows the convergence of the value $V^{(k)}(s_0)$ in APT-DIFT game with the iteration number $k$. The threshold value of the maximum absolute error, $\delta$, in Algorithm~\ref{alg:value} was set to $10^{-7}$. We note that $\delta= \max_{s \in {\bf S}} |V^{(k)}(s) - V^{(k-1)}(s)|$ where $k = 1, 2, \ldots$. Figure~\ref{fig:Algo_one}(b) shows that $\delta$ monotonically decreases with $k$. At iteration $k = 32$,  $\delta = 7.79 \times 10^{-8}$.

HSL algorithm (Algorithm~\ref{alg:twostage}) was tested on APT-DIFT game associated with the pruned acyclic IFG of ransomware attack given in Figure~\ref{fig:IFG_ransomware} and results are given in Figure~\ref{fig:Algo_Two}. We used a sequential neural network with two dense layers to learn the value vector for a given strategy pair of APT and DIFT. Each dense layer consists of $1000$ neurons and ReLU activation function. Training data consist tuples of APT and DIFT strategy pair and corresponding value vector of APT-DIFT game.  We generated $1 \times 10^{5}$ training data samples that consist of randomly generated deterministic APT policies. We randomly generated DIFT's policies such that $40\%$ strategies are stochastic (mixed) and $60\%$ are deterministic. For each randomly generated APT and DIFT strategy pair, the corresponding value vector was computed. A stochastic gradient descent optimizer was used to train the neural network. In each experiment trial the neural network  was trained for $100$ episodes and validation error was set to $< 1\%$.


In Figure~\ref{fig:Algo_Two}(a) we analyze the sensitivity of estimated value vector to the variations of payoff parameter, $\beta$ by plotting the mean absolute error $\mu$ between actual value vector $V$ and estimated value vector $\hat{V}$, i.e., $\mu = \sum_{s \in {\bf S}} |V(s) - \hat{V}(s)| / |{\bf S}|$, with respect to $\beta$.  In these experiments $10^{5}$ training data samples were used in HSL algorithm. 
The results show that $\hat{V}$ values are close and consistent with $V$ values when $\beta$ parameter takes smaller values and variations between $\hat{V}$ and $V$ is increased when $\beta$ takes larger values. A reason for this behavior can be the numerical unstability associated with the training and estimating tasks done in the neural network model used in HSL algorithm.
In order to study the effect of the length of training data samples on the estimated value vector we plot $\mu$ against number of training data samples used in HSL algorithm in Figure~\ref{fig:Algo_Two}(b). $\beta = 50$ was used in the experiments. The results show that $\mu$ decreases with the number of training data samples used in HSL algorithm as the increased number of training data samples improves the learning in  neural network model.



\section{Conclusion}\label{sec:end}
This paper studied detection of Advanced Persistent Threats (APTs) using a  Dynamic Information Flow Tracking (DIFT)-based detection mechanism. We modeled the strategic interaction between the APT and DIFT as a non-cooperative stochastic game with total reward structure. The APT-DIFT game has  imperfect information as both APT and DIFT are unaware of the actions of the opponent. Also, our game model incorporates the false-positives  and false-negatives generated by DIFT.  We considered two scenarios of the game~(i)~ when the transition probabilities, i.e., rates of false-positives  and false-negatives, are known to both players and (ii)~ when the transition probabilities are unknown to both players. For case~(i), we proposed a value iteration-based algorithm and proved convergence and complexity. For case~(ii), we proposed  Hierarchical Supervised Learning (HSL),  a supervised learning-based algorithm that utilizes the hierarchical structure of the state space of the APT-DIFT game, that integrates a neural network with a policy iteration algorithm to compute an approximate equilibrium of the game. We implemented our algorithms on real attack datasets for nation state and ransomware attacks collected using the RAIN framework and validated our results and performance.



\bibliographystyle{abbrv}
\bibliography{DIFT-Reference}

\section{Appendix}

{\em Proof~of~Proposition~\ref{prop:NE}}:
It is known that the class of zero-sum, finite, stochastic games with nonzero stopping (termination) probability has  Nash equilibrium in stationary strategies \cite{shapley}. Note that, the APT-DIFT game is a stochastic game with finite state space and finite action spaces for both players. Moreover, the transition probabilities in the game $FP(\cdot)$ and $FN(\cdot)$ are such that $0 < FP(\cdot)<1$ and $0< FN(\cdot)< 1$.
Thus the stopping probabilities of the APT-DIFT game are nonzero, for all strategy pairs $(p_{\sA}, p_{\sD})$ except for a deterministic policy in which the defender does not perform security analysis at any state. However, such a policy is trivially irrelevant to the game as the defender is idle and essentially not participating in the game.   As the result for zero-sum games also hold for constant-sum games,  from \cite{shapley} it follows that there exists an NE for the APT-DIFT game.
\qed

\end{document}